\documentclass[smallabstract,smallcaptions]{dccpaper}
\usepackage[utf8]{inputenc}

\usepackage{graphicx}
\usepackage[english]{babel}
\usepackage{amsmath}
\usepackage{amsthm}
\usepackage{amssymb}
\usepackage[boxed,vlined,linesnumbered]{algorithm2e}
\usepackage{epsfig}
\usepackage{forest}
\usepackage{color}
\usepackage{url}
\usepackage{tikz}
\usepackage{enumerate}
\usepackage{epsfig}
\usepackage{float}
\usepackage{mathtools}

\usepackage{pgfplots}
\usepackage{pgfplotstable}
\usepgfplotslibrary{statistics}
\usetikzlibrary{backgrounds,calc,patterns,positioning,shadings}
\pgfplotsset{compat=1.13}

\newtheoremstyle{subdefistyle}
{0pt} 
{0pt} 
{} 
{} 
{\bfseries} 
{)} 
{.5em} 
{} 

\newtheoremstyle{sublemmastyle}
{0pt} 
{0pt} 
{\it} 
{} 
{\bfseries} 
{)} 
{.5em} 
{} 

\theoremstyle{definition}
\newtheorem{defi}{Definition}

\newtheorem{problem}[defi]{Problem}

\theoremstyle{plain}
\newtheorem{theorem}[defi]{Theorem}

\newtheorem{lemma}[]{Lemma}
\newtheorem{cor}[defi]{Corollary}

\DeclareMathOperator*{\argmin}{argmin}

\theoremstyle{subdefistyle}

\theoremstyle{sublemmastyle}

\usepackage{multirow}

\newcommand{\sdots}{..}
\newcommand{\llex}{<_{\mathsf{lex}}}
\newcommand{\kmers}{\mathcal{K}}
\newcommand{\sufftrie}[1][]{RT_{#1}}
\newcommand{\wlinks}{WL}

\newcommand{\LCol}{\mathsf{L}}
\newcommand{\BWT}{\mathsf{L}}
\newcommand{\LF}{\mathsf{LF}}
\newcommand{\getintervals}{\mathsf{getIntervals}}
\newcommand{\C}{\mathsf{C}}
\newcommand{\rank}{\mathsf{rank}}
\newcommand{\select}{\mathsf{select}}

\newcommand{\dout}{D_{\mathsf{out}}}
\newcommand{\din}{D_{\mathsf{in}}}

\newlength{\figurewidth}
\newlength{\smallfigurewidth}

\title
{
\large
\textbf{Edge minimization in de Bruijn graphs}
}

\author{%
Uwe Baier$^{\ast}$, Thomas B\"{u}chler$^{\ast}$, Enno Ohlebusch$^{\ast}$, Pascal Weber$^{\ast}$\\[0.5em]
{\small\begin{minipage}{\linewidth}\begin{center}
$^{\ast}$Institute of Theoretical Computer Science \\
Ulm University, D-89069 Ulm, Germany \\
\url{{uwe.baier,thomas.buechler,enno.ohlebusch,pascal-1.weber}@uni-ulm.de}
\end{center}\end{minipage}}
}

\begin{document}

\maketitle
\thispagestyle{empty}

\begin{abstract}
This paper introduces the de Bruijn graph edge minimization problem, which is 
related to the compression of de Bruijn graphs:
find the order-$k$ de Bruijn graph with minimum edge count among all orders.
We describe an efficient algorithm that solves this problem.
Since the edge minimization problem is connected to the BWT compression 
technique called ``tunneling'', the paper also describes a way to minimize 
the length of a tunneled BWT in such a way that useful properties for
sequence analysis are preserved. Although being a restriction,
this is significant progress towards a solution to the open problem
of finding optimal disjoint blocks that minimize space, as stated in Alanko et al.\ (DCC 2019).
\end{abstract}


\Section{Introduction}

De Bruijn graphs play an important role in string processing. Originally
invented for solving combinatorial problems \cite{DEBRU:1946}, de
Bruijn graphs have also been used in bioinformatics, e.g.\ in genome assembly \cite{IDU:WAT:1995} or to 
describe variations between different strings \cite{IQB:CAC:TUR:FLI:MCV:2012}. As de Bruijn graphs 
that show variations between genomes can get very large, there have been attempts to
compress such graphs by merging nodes \cite{MAR:LEE:SCH:2014}.

In this paper, we will present a similar compression scheme for such graphs with the
difference that we view the de Bruijn graph as a multigraph in which every
overlapping of two consecutive $k$-mers in the underlying cyclic string induces an edge in the graph. Then, instead
of merging nodes, we fuse the edges between two nodes $x$ and $y$ if $y$ is the only successor
of $x$ and $x$ is the only predecessor of $y$. We also introduce the de Bruijn graph edge
minimization problem, asking for the order $k$ such that an edge-reduced de Bruijn graph of an
underlying string has the minimum amount of edges under all possible orders of the graph.

The main contribution of our paper is the presentation of an algorithm which, given a
special trie, is capable of solving the de Bruijn graph edge minimization
problem with a worst-case time that is linear on the minimum amount of edges. The algorithm is
designed in such a way that it can be implemented with an FM-index \cite{BUR:WHE:1994,FER:MAN:2005}
in an overall run-time of $O(n \log \sigma)$, where $\sigma$ is the alphabet size.

We also show a direct connection between edge reduction in de Bruijn graphs and a recent BWT
compression technique called tunneling \cite{BAI:2018}. It will be shown that each edge-reduced de Bruijn
graph corresponds to a tunneled BWT of the underlying string, where the number of edges in the
graph is identical to the length of the tunneled BWT.  Therefore, we show that solving the edge minimization
problem provides significant progress towards a solution to the open problem of finding the optimal disjoint blocks
that minimize space, as stated in \cite{ALA:GAG:NAV:BEN:2019}.

Finally, the paper describes how the outputs of the edge minimization algorithm can be used to
produce the associated tunneled BWT. We provide a publicly available implementation of
our algorithms \cite{edgemin-impl}
and present results for the edge reduction ratio and tunneled BWT size ratio for real world
data coming from the Pizza \& Chili text corpus \cite{pizzachili-corpus}
 and the Repetitive corpus \cite{repetitive-corpus}.


\Section{Preliminaries}
We start with some standard string notation. Throughout this paper, any interval
$[i,j]$ is meant to be an interval over the natural numbers and indices
start with $1$, except when stated differently.

Let $\Sigma$ be an ordered alphabet of size $\sigma = |\Sigma|$.
A string $S$ of length $n=|S|$ over $\Sigma$ is a finite sequence of $n$ 
characters from $\Sigma$. Throughout this paper, we assume that $S$ is
null-terminated, i.e.\ that it ends with the EOF-symbol $\$$ 
(which is the smallest symbol in $\Sigma$).
For $i,j \in [1,n]$, $S[i]$ denotes the $i$-th character of $S$ and
$S[i\sdots j]$ denotes the substring of $S$ starting at the $i$-th and ending 
at the $j$-th position. A length $k$ substring of $S$ is called a $k$-mer.
The empty string with length $0$ is denoted by $\varepsilon$. 
Let $S$ and $T$ be two strings of length $n$ and $m$ over a
common alphabet $\Sigma$. We write $S \llex T$ if $S$ is lexicographically 
smaller than $T$, i.e.\ $S$ is a proper prefix of $T$ or there exists a
$k \in [1,\min\{n,m\}]$ with $S[1\sdots k-1] = T[1\sdots k-1]$ and 
$S[k] < T[k]$.

Next, the de Bruijn graph (DBG) shall be defined. 
Note that our definition differs from
common definitions as \cite{COM:PEV:TES:2011} or \cite{BAI:BEL:OHL:2016} because we consider
de Bruijn graphs originating from cyclic strings.

\begin{defi} \label{def:debruijngraph}
Let $S$ be a string of length $n$ and $k \in [1,n]$. If we concatenate
the $k$-mer prefix of $S$ to $S$ itself, then we obtain the string
$Z_k(S) \coloneqq S[1\sdots n]S[1\sdots k]$.
Furthermore, the set of all $k$-mers of $Z_k(S)$ is
\[
	\kmers \coloneqq \{ ~ Z_k(S)[i\sdots i+k-1] ~ \mid ~ i \in [1,n] ~ \}\text{.}
\]
The de Bruijn graph $G_k(S) = (\kmers, E)$ of order $k$ is a directed multigraph, where the multiset of edges is defined as (the superscript $m$ denotes the multiplicity of an edge):
\[
	E \coloneqq \{ ~ (x[1\sdots k],x[2\sdots k+1])^m ~ \mid
	               ~ x \in \Sigma^{k+1} \text{ occurs exactly } m \text{ times in } Z_k(S) \}\text{.}
\]
\end{defi}

An example of a DBG can be found in Figure \ref{fig:dbg-st}. It should be noted that,
independent of the order $k$, each DBG $G_k(S)$ contains exactly $|S|$ edges due to the definition of
the string $Z_k(S)$. The next definition introduces the rotation-trie
$\sufftrie(S)$ built from all rotations (cyclic shifts) of $S$.

\begin{defi} \label{def:sufftrie}
Let $S$ be a string of length $n$. The rotation-trie $\sufftrie(S)$ is a trie \cite{FRE:1960} 
built from the rotations $S[1\sdots n] , S[2\sdots n]S[1] , \ldots , S[n]S[1\sdots n-1]$
of $S$ such that the following holds for any two nodes $\phi$ and $\varphi$ of the trie: If $\phi$ is the left sibling of $\varphi$,
then its label%
\footnote{The label of a node $\phi$ in a trie emerges by concatenating all edge labels in the path
	from the root node to node $\phi$.}
$l(\phi)$ is lexicographically smaller than the label of $\varphi$, i.e.\ $l(\phi) \llex l(\varphi)$.

We furthermore denote by $lb(\phi)$ resp.\ $rb(\phi)$ the lexicographic rank of the leftmost resp.\
rightmost leaf of $\phi$ among all leaves of $\sufftrie(S)$, and by $|\phi|$ the number of leaves of the subtree rooted at $\phi$,
i.e.\ $|\phi| = rb(\phi) - lb(\phi) + 1$. With $\sufftrie[k](S)$ we denote the set of 
nodes at level $k$ in $\sufftrie(S)$. $\sufftrie[\leq k](S)$ denotes the nodes in the levels $1$ to $k$.
 
\end{defi}

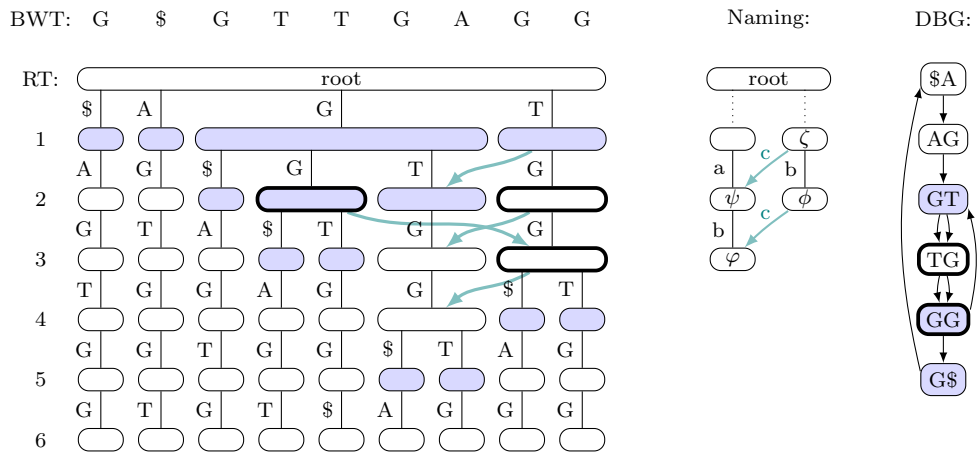
\begin{figure}[t]
\centering
\tikzset{>=latex}
\begin{tikzpicture}[scale=0.8,
	one/.style={minimum width=2em},
	two/.style={minimum width=4.8em},
	five/.style={minimum width=13em},
	six/.style={minimum width=15.5em},
	stnode/.style={draw, rounded corners, minimum height=1em, one, fill=blue!15},
	dbgnode/.style={draw, rounded corners, minimum height=1em, one, fill=blue!15},
	namingnode/.style={draw, rounded corners, minimum height=1em, one, fill=blue!0},
	white/.style={fill=white},
	f/.style={line width=0.5mm},
	wlink/.style={->,very thick,teal!50}]
\begin{scriptsize}
    \def \anchor{0}; 
    \def \anchordbg{15}; 
    \def \anchornaming{11.5};

    \node at (\anchor,7) {BWT:};
    \node at (\anchor,6) {RT:};
    \node at (\anchor,5) {1};
    \node at (\anchor,4) {2};
    \node at (\anchor,3) {3};	
    \node at (\anchor,2) {4};
    \node at (\anchor,1) {5};
    \node at (\anchor,0) {6};	

    \node  at (\anchor+1,7) (){G};
    \node  at (\anchor+2,7) (){\$};
    \node  at (\anchor+3,7) (){G};
    \node  at (\anchor+4,7) (){T};
    \node  at (\anchor+5,7) (){T};
    \node  at (\anchor+6,7) (){G};
    \node  at (\anchor+7,7) (){A};
    \node  at (\anchor+8,7) (){G};
    \node  at (\anchor+9,7) (){G};

    \node [stnode,minimum width=23.5em,white] at (\anchor+5,6) (0) {};
    \node [] at (\anchor+5,6)  {root};
    
    \node [stnode] 		at (\anchor+1,5) 	(!) {};
    \node [stnode] 		at (\anchor+2,5) 	(A) {};
    \node [stnode,five] at (\anchor+5,5)	(G) {};
    \node [stnode,two] 	at (\anchor+8.5,5)	(T) {};
    
    \node [stnode,white] 	at (\anchor+1,4) 	(!A){};
    \node [stnode,white] 			at (\anchor+2,4) 	(AG){};
    \node [stnode] 			at (\anchor+3,4) 	(G!){};
    \node [stnode,two,f] 	at (\anchor+4.5,4)	(GG){};
    \node [stnode,two] 		at (\anchor+6.5,4) 	(GT){};
    \node [stnode,two,white,f] at (\anchor+8.5,4)	(TG){};  
    
    \node [stnode,white] 	at (\anchor+1,3) (!AG){};
    \node [stnode,white] 	at (\anchor+2,3) (AGT){};
    \node [stnode,white] 	at (\anchor+3,3) (G!A){};
    \node [stnode] 			at (\anchor+4,3) (GG!){};
    \node [stnode] 			at (\anchor+5,3) (GGT){};
    \node [stnode,two,white] 	at (\anchor+6.5,3) 	(GTG){};
    \node [stnode,two,white,f] 	at (\anchor+8.5,3) (TGG){}; 
      
    \node [stnode,white] at (\anchor+1,2) (!AGT){};
    \node [stnode,white] at (\anchor+2,2) (AGTG){};
    \node [stnode,white] at (\anchor+3,2) (G!AG){};
    \node [stnode,white] at (\anchor+4,2) (GG!A){};
    \node [stnode,white] at (\anchor+5,2) (GGTG){};
    \node [stnode,two,white] 	at (\anchor+6.5,2) (GTGG){};
    \node [stnode] 				at (\anchor+8,2) (TGG!){};
    \node [stnode] 				at (\anchor+9,2) (TGGT){};

	\node [stnode,white] at (\anchor+1,1) (!AGTG){};
    \node [stnode,white] at (\anchor+2,1) (AGTGG){};
    \node [stnode,white] at (\anchor+3,1) (G!AGT){};
    \node [stnode,white] at (\anchor+4,1) (GG!AG){};
    \node [stnode,white] at (\anchor+5,1) (GGTGG){};
    \node [stnode] at (\anchor+6,1) (GTGG!){};
    \node [stnode] at (\anchor+7,1) (GTGGT){};
    \node [stnode,white] at (\anchor+8,1) (TGG!A){};
    \node [stnode,white] at (\anchor+9,1) (TGGTG){};

	\node [stnode,white] at (\anchor+1,0) (!AGTGG){};
    \node [stnode,white] at (\anchor+2,0) (AGTGGT){};
    \node [stnode,white] at (\anchor+3,0) (G!AGTG){};
    \node [stnode,white] at (\anchor+4,0) (GG!AGT){};
    \node [stnode,white] at (\anchor+5,0) (GGTGG!){};
    \node [stnode,white] at (\anchor+6,0) (GTGG!A){};
    \node [stnode,white] at (\anchor+7,0) (GTGGTG){};
    \node [stnode,white] at (\anchor+8,0) (TGG!AG){};
    \node [stnode,white] at (\anchor+9,0) (TGGTGG){};

    \draw (!)+(0,0.82) -- node[left] {\$} (!.north);
    \draw (A)+(0,0.82) -- node[left] {A} (A.north);
    \draw (G)+(0,0.82) -- node[left] {G} (G.north);
    \draw (T)+(0,0.82) -- node[left] {T} (T.north);
    
    \draw (!) -- node[left] {A} (!A.north);
    \draw (A) -- node[left] {G} (AG.north);
    \draw (G!)+(0,0.82) -- node[left] {\$} (G!.north);
    \draw (GG)+(0,0.82) -- node[left] {G} (GG.north);
    \draw (GT)+(0,0.82) -- node[left] {T} (GT.north);
    \draw (T) -- node[left] {G} (TG.north);
    
    \draw (!A) -- node[left] {G} (!AG.north);
    \draw (AG) -- node[left] {T} (AGT.north);
    \draw (G!) -- node[left] {A} (G!A.north);
    \draw (GG!)+(0,0.82) -- node[left] {\$} (GG!.north);
    \draw (GGT)+(0,0.82) -- node[left] {T} (GGT.north);
    \draw (GT) -- node[left] {G} (GTG.north);
    \draw (TG) -- node[left] {G} (TGG.north);
    
    \draw (!AG) -- node[left] {T} (!AGT.north);
    \draw (AGT) -- node[left] {G} (AGTG.north);
    \draw (G!A) -- node[left] {G} (G!AG.north);
    \draw (GG!) -- node[left] {A} (GG!A.north);
    \draw (GGT) -- node[left] {G} (GGTG.north);
    \draw (GTG) -- node[left] {G} (GTGG.north);
    \draw (TGG!)+(0,0.82) -- node[left] {\$} (TGG!.north);
    \draw (TGGT)+(0,0.82) -- node[left] {T} (TGGT.north);
    
    \draw (!AGT) -- node[left] {G} (!AGTG.north);
    \draw (AGTG) -- node[left] {G} (AGTGG.north);
    \draw (G!AG) -- node[left] {T} (G!AGT.north);
    \draw (GG!A) -- node[left] {G} (GG!AG.north);
    \draw (GGTG) -- node[left] {G} (GGTGG.north);
    \draw (GTGG!)+(0,0.82) -- node[left] {\$} (GTGG!.north);
    \draw (GTGGT)+(0,0.82) -- node[left] {T} (GTGGT.north);
    \draw (TGG!)-- node[left] {A} (TGG!A.north);
    \draw (TGGT)-- node[left] {G} (TGGTG.north);
    
    \draw (!AGTG) -- node[left] {G} (!AGTGG.north);
    \draw (AGTGG) -- node[left] {T} (AGTGGT.north);
    \draw (G!AGT) -- node[left] {G} (G!AGTG.north);
    \draw (GG!AG) -- node[left] {T} (GG!AGT.north);
    \draw (GGTGG) -- node[left] {\$} (GGTGG!.north);
    \draw (GTGG!) -- node[left] {A} (GTGG!A.north);
    \draw (GTGGT) -- node[left] {G} (GTGGTG.north);
    \draw (TGG!A) -- node[left] {G} (TGG!AG.north);
    \draw (TGGTG) -- node[left] {G} (TGGTGG.north);
    
	\begin{scope}[on background layer]
	\draw[wlink] (GG)  to[out=-20,in=150] (TGG);
	\draw[wlink] (T)  to[out=-150,in=40] (GT);
	\draw[wlink] (TG)  to[out=-150,in=40] (GTG);
	\draw[wlink] (TGG)  to[out=-150,in=40] (GTGG);
	
	\end{scope}
    

	\node at (\anchordbg,7) {DBG:};

	\node [dbgnode,white] 	at (\anchordbg,6) 	(!A) {\$A};	
	\node [dbgnode,white] 		at (\anchordbg,5)	(AG) {AG};
	\node [dbgnode] 		at (\anchordbg,4) 	(GT) {GT};
	\node [dbgnode,f,white] at (\anchordbg,3) 	(TG) {TG};
	\node [dbgnode,f]     	at (\anchordbg,2) 	(GG) {GG};
	\node [dbgnode] 		at (\anchordbg,1)   (G!) {G\$};
	
	\path[->]
	(!A) edge (AG)
	(AG) edge (GT) 
	(GG) edge (G!)
	[bend left=13]
	(GT) edge (TG)
	(TG) edge (GG)
	(G!.160) edge (!A.200)
	[bend left=-13]
	(GG.20) edge (GT.-20)
	(GT) edge (TG)
	(TG) edge (GG);


	\node at (\anchornaming+.6,7) {Naming:};

    \node [namingnode,minimum width= 5.5em] at (\anchornaming+.6,6) (0) {};  
    \node  at (\anchornaming+.6,6)  {root};  
    
    \node [namingnode] at (\anchornaming,5)		(1){};
    \node [namingnode] at (\anchornaming+1.2,5)	(2){};  
    
    \node [namingnode] at (\anchornaming,4) 	(3){};    
    \node [namingnode] at (\anchornaming+1.2,4) (4){};
    
    \node [namingnode] at (\anchornaming,3) 	(5){};    
    
    
    \node at (2) () {$\zeta$};
    \node at (4) () {$\phi$};
    
    \node at (3) () {$\psi$};
    \node at (5) () {$\varphi$};

    \draw [dotted] (1)+(0,.82) -- node[left] {} (1);
    \draw [dotted] (2)+(0,.82) -- node[left] {} (2);
    
    \draw (1) -- node[left] {a} (3);
    \draw (2) -- node[left] {b} (4);   
     
    \draw (3) -- node[left] {b} (5);
        
    
    \draw[wlink,thick] (2)  to[out=-145,in=45] node[above,very thick,teal] {c} (3);
    \draw[wlink,thick] (4)  to[out=-145,in=45] node[above,very thick,teal] {c} (5);

\end{scriptsize}
\end{tikzpicture}
\caption{
Left: First 6 levels of the rotation-trie $\sufftrie(S)$ of the string $S=$ AGTGGTGG\$.
Nodes with siblings are filled. Teal-colored arrows represent unique Weiner links
(with the connected nodes having size $>1$). The nodes of which incoming edges can be 
fused in the DBG are drawn bold-framed.
Middle: Greek letters for nodes are used as depicted here.
Right: Order-$2$ de Bruijn graph $G_2(S)$. 	
}
\label{fig:dbg-st}
\end{figure}

Figure \ref{fig:dbg-st} shows an example of a rotation-trie. Note that all the
rotations of a null-terminated string are distinct, which implies that
$\sufftrie(S)$ has $n=|S|$ leaves.
To see the connection between RTs and DBGs we need the definition of so-called 
Weiner links \cite{WEI:1973}.

\begin{defi} \label{def:wlinks}
Let $\sufftrie(S)$ be a rotation-trie and let $\phi$ be a node of 
$\sufftrie(S)$. We call the set
\begin{align*}
\wlinks  \coloneqq & \{ ~ (\phi,\varphi,c) ~ \mid ~ \varphi \text{ is a node of } \sufftrie(S)
	                         \text{ such that for some } c \in \Sigma \\
                   &\phantom{ \{ ~ (\phi,\varphi,c) ~ \mid ~ } \, l(\varphi) = c l(\phi) \text{ or }
		|l(\phi)| = |S| \text{ and } l(\varphi) = cl(\phi)[1\sdots|S|-1] ~ \} 
\end{align*}
the set of Weiner links of node $\phi$. A Weiner link $(\phi,\varphi,c) \in \wlinks$ is called unique
if there exists no other Weiner link $(\phi,\psi,d) \in \wlinks$ with $\varphi \neq \psi$.
\end{defi}

We state the following connections between de Bruijn graphs and rotation-tries:
let $G_k(S) = (\kmers, E)$ be a 
DBG of order $k$ for a string $S$, and let $\sufftrie(S)$
be the rotation-trie of the same string $S$. It can be seen that the $k$-mers in 
$\kmers$ correspond to the nodes in $\sufftrie[k](S)$, i.e.\ 
for every node $y \in \kmers$ there is exactly one node $\phi$ in
$\sufftrie[k](S)$ with $l(\phi) = y$.

\begin{lemma} \label{lem-edgeWL}
Let $k < n$.
There is an edge $(x,y)^m \in E$ in $G_k(S) = (\kmers, E)$ with $c=x[1]$
if and only if there is a Weiner link $(\phi,\varphi,c)$ in $\sufftrie(S)$
such that
(a) $l(\phi) = y$, (b) $l(\varphi) = cy$, and (c) $m$ is the number of leaves
$\vartheta$ in the subtree rooted at $\phi$ with $l(\vartheta)[n]=c$
(i.e.\ the last character in the label of the leaf is $c$).
\end{lemma}
\begin{proof}
'$\Rightarrow$' If $(x,y)^m \in E$ and $c=x[1]$, then by definition the
string $z=cy$ (of length $k+1$) occurs exactly $m$ times in $Z_k(S)$.
Obviously, there are nodes $\phi$ and $\varphi$ in $\sufftrie(S)$ such that
$l(\phi) = y$ and $l(\varphi) = cy$. Hence there is Weiner link 
$(\phi,\varphi,c)$. Consider a leaf $\vartheta$ in the subtree rooted at 
$\phi$ with $l(\vartheta)[n]=c$. The string $l(\vartheta)$ is 
a rotation of $S$. If we cyclically shift the last character $c$ of 
$l(\vartheta)$ to the front, then we obtain another rotation of $S$ that has 
$cy$ as a prefix. It is not difficult to see that there are $m$ occurrences 
of $cy$ in $Z_k(S)$ if and only if there a $m$ rotations of $S$ that have 
$cy$ as a prefix. This proves the claim.\\
'$\Leftarrow$' Suppose there is a Weiner link 
$(\phi,\varphi,c)$ in $\sufftrie(S)$
such that (a) $l(\phi) = y$, (b) $l(\varphi) = cy$, and (c) $m$ is the number 
of leaves in the subtree rooted at $\phi$ with $l(\vartheta)[n]=c$. 
This implies that the string $z=cy$ occurs exactly $m$ times in $Z_k(S)$.
Thus, $(x,y)^m \in E$.
\end{proof}


\Section{Edge minimization in de Bruijn graphs}

In this section, we will introduce the DBG edge minimization problem.
Edge reduction between two nodes $x,y \in \kmers$ in a DBG 
$G_k(S)$ is defined as follows: if $y$ is the only
successor of $x$ and $x$ is the only predecessor of $y$, 
then the number of edges between $x$ and $y$ can be reduced to one.

\begin{defi} \label{def:edgereduceddbg}
Let $G_k(S)=(\kmers,E)$ be a DBG for string $S$ of order $k$. Denote by
\[
	F \coloneqq \{ ~ (x,y)^m \in E ~ \mid ~ y \text{ is the only successor of } x \text{ and } x \text{ is the only predecessor of } y ~ \}
\]
the set of fusible edges. The edge-reduced de Bruijn graph ${\tilde G}_k(S)=(\kmers,\tilde E)$ of $G_k(S)$
is defined by decreasing the multiplicity of all fusible edges to one, that is,
\[
	\tilde E \coloneqq ( E \setminus F ) \uplus \{ ~ (x,y)^1 ~ \mid ~ (x,y)^m \in F ~ \}
\]
\end{defi}

An example of edge-reduced DBGs can be found in Figure \ref{fig:dbgs-edgereduced}. Definition
\ref{def:edgereduceddbg} is closely related to so-called compressed de Bruijn graphs \cite{MAR:LEE:SCH:2014},
with the difference that in Definition \ref{def:edgereduceddbg} edges are fused while in compressed DBGs nodes are merged.
We are now able to state the edge minimization problem.

\begin{problem}
Given a string $S$ of length $n$, the de Bruijn graph edge minimization problem asks for
finding an order $k \in [1,n]$ such that the amount of edges in the edge-reduced DBG
${\tilde G}_k(S)=(\kmers,\tilde E)$ is minimal, that is, find the order $k$ such that
\[
	k = \argmin_{\tilde k \in [1,n]}\{ ~ \sum_{(x,y)^m \in \tilde E} m ~ \mid ~ {\tilde G}_{\tilde k}(S) = (\kmers,{\tilde E})
	\text{ is the order } \tilde k \text{ edge-reduced DBG of } S ~ \}
\]
\end{problem}

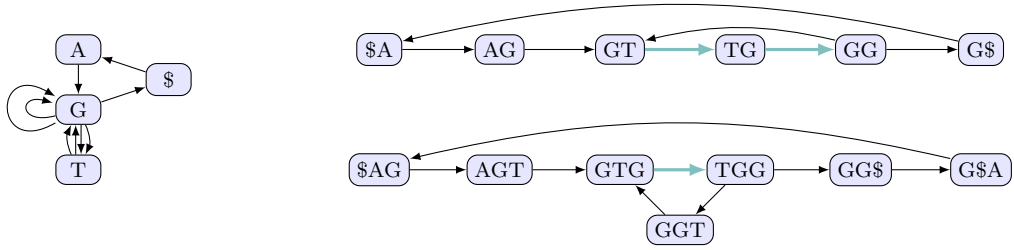
\begin{figure}[t]
\centering
\tikzset{>=latex}
\begin{tikzpicture}[scale=0.8,
	one/.style={minimum width=2em},
	dbgnode/.style={draw, rounded corners, minimum height=1em, one, fill=blue!10}]
\begin{scriptsize}

//k=1    
\node[dbgnode] at (2,0) 	(A) {A};
\node[dbgnode] at (3.5,-0.5) 	(!) {\$};
\node[dbgnode] at (2,-1)		(G) {G};
\node[dbgnode] at (2,-2)	(T) {T};

\path[->]
(!) edge (A)
(A) edge (G)
(G) edge (!)
(G.-80) edge (T.80)
(T.100) edge (G.-100)
(G) [loop left] edge (G)
(G) [loop left,in=150,out=210,looseness=7] edge (G);

\path[->,bend left=25]
(G) edge (T)
(T) edge (G);

//k=2
\node [dbgnode] 	at (7,0) 	(!A) {\$A};	
\node [dbgnode] 	at (9,0)	(AG) {AG};
\node [dbgnode] 	at (11,0) 	(GT) {GT};
\node [dbgnode] 	at (13,0) 	(TG) {TG};
\node [dbgnode]    	at (15,0) 	(GG) {GG};
\node [dbgnode]		at (17,0)  	(G!) {G\$};
	
\path[->]
(!A) edge (AG)
(AG) edge (GT) 
(GG) edge (G!)
[bend left=-13]
(G!.160) edge (!A.20)
(GG.160) edge (GT.20);
	
\path[->,very thick,teal!50]
(GT) edge (TG)
(TG) edge (GG);

//k=3;

\node[dbgnode] 	at (7	,-2) 	(!AG){\$AG};
\node[dbgnode] 	at (9	,-2) 	(AGT){AGT};
\node[dbgnode] 	at (11	,-2) 	(GTG){GTG};
\node[dbgnode] 	at (13	,-2) 	(TGG){TGG};
\node[dbgnode] 	at (15	,-2) 	(GG!){GG\$};
\node[dbgnode] 	at (17,	-2) 	(G!A){G\$A};

\node[dbgnode] 	at (12	,-3) 	(GGT){GGT};

\path[->]
(!AG) edge (AGT)
(AGT) edge (GTG)
(TGG) edge (GG!)
(GG!) edge (G!A)
(TGG) edge (GGT)
(GGT) edge (GTG)
[bend left=-13]
(G!A.160) edge (!AG.20);

\path[->,very thick,teal!50]
(GTG) edge (TGG);

\end{scriptsize}
\end{tikzpicture}
\caption{Edge-reduced de Bruijn graphs with order $k=1$ (left), $k=2$ (upper right) and $k=3$ (lower right)
	built from the string AGTGGTGG\$.
	Fused edges are indicated with teal arrows, the DBG with order $k=2$ contains the least edges,
	namely $9-2 = 7$ edges.}
\label{fig:dbgs-edgereduced}
\end{figure}

Figure \ref{fig:dbgs-edgereduced} illustrates the problem. At first sight, it seems that 
the problem is easily solvable: The higher the order gets, the more nodes exist,
resulting in fewer branching edges and therefore in more edge fusions. Then, at
a high enough order $k$, the amount of nodes prevents edge fusions because there exist not enough edges between nodes. 
Summing these observations up, one could think that the amount of
edges in an edge-reduced DBG follows a parabolic-shaped function depending on the order $k$,
where the only local minimum coincides with the global minimum.
Unfortunately, as Figure
\ref{fig:k-edgecnt-matrix} shows, this observation does not hold true for all strings: the file
\texttt{influenza} clearly shows that the function can contain multiple local minima,
so we need more advanced techniques to compute the order $k$ such that
${\tilde G}_k(S)$ has a minimum amount of edges.

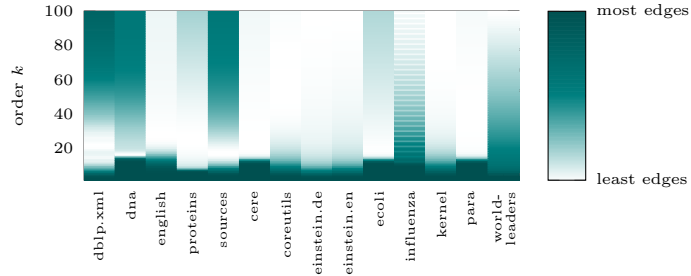
\begin{figure}[t]
\centering
\begin{tikzpicture}
\begin{axis}[y=.125ex,x=1em,ylabel={\tiny order $k$},
             colormap={redyellow}{rgb255(0cm)=(255,255,255); rgb255(1cm)=(0,128,128); rgb255(2cm)=(0,80,80)},
             colorbar,
             colorbar style={
		ytick={0,1},
		yticklabels={{\tiny least edges},{\tiny most edges}}
             },
	     xticklabel={\pgfmathparse{int(\tick-1)}\pgfplotstablegetelem{\pgfmathresult}{File}\of{edgeinfo.dat}\tiny\rotatebox{90}{\fontsize{.15cm}{.15cm}\selectfont\pgfplotsretval}},
             yticklabel style={font=\tiny},
	     xtick={1,...,14},
             enlargelimits=false,
             xmin=.5,xmax=14.5,
             ymin=1,ymax=100
             ]
	\addplot [matrix plot*,point meta=explicit] file [meta=index 2] {edgematrix.dat};
\end{axis}
\end{tikzpicture}
\vspace{-.5\baselineskip}
\caption{Matrix plot showing for each test file the relative amount of edges depending on the order
	$k$ of the edge-reduced DBG.}
\label{fig:k-edgecnt-matrix}
\end{figure}

\SubSection{Naive computational approach}

It is a consequence of Lemma \ref{lem-edgeWL} that 
the set of predecessors of a node $y$ in the graph $G_k(S)$
(the set of nodes $x\in \kmers$ with $(x,y)^m \in E$)
 can be determined as follows:
Find the node $\phi$ in $\sufftrie(S)$ that corresponds to $y$ (i.e.\
$l(\phi) = y$) and follow its Weiner links. For each Weiner link 
$(\phi,\varphi,c)$, the parent node $\psi$ of $\varphi$ is a predecessor
of $x$ in $G_k(S)$ (more precisely, $x=l(\psi)$ is a predecessor).
The next lemma is the basis of our algorithms.

\begin{lemma} \label{lem-key}
A node $\phi$ in $\sufftrie(S)$ has a unique Weiner link 
$(\phi,\varphi,c)$ and $\varphi$ has no siblings if and only if 
$x =l(\psi)$ is the only predecessor of $y=l(\phi)$, where
$\psi$ is the parent node of $\varphi$ in $G_k(S)$,
and $y$ is the only successor of $x$.
\end{lemma}
\begin{proof}
It follows directly from Lemma \ref{lem-edgeWL} that $\phi$ has a unique 
Weiner link $(\phi,\varphi,c)$ if and only if 
$x =l(\psi)$ is the only predecessor of $y=l(\phi)$ in $G_k(S)$, where
$\psi$ is the parent node of $\varphi$ in $\sufftrie(S)$.
Suppose that $\phi$ has a unique Weiner link $(\phi,\varphi,c)$.
We must show that $\varphi$ has no siblings (or equivalently,
its parent $\psi$ has only one child) if and only if
$y=l(\phi)$ is the only successor of $x=l(\psi)$.
Let $b = y[k]$ be the last character of $y$.\\
'$\Rightarrow$' If $\varphi$ has no siblings, then every occurrence of $x$
in $Z_k(S)$ is followed by the character $b$. Therefore,
$y$ is the only successor of $x$.\\
'$\Leftarrow$' If $\varphi$ has a sibling $\vartheta$, then there is
an occurrence of $x$ in $Z_k(S)$ that is followed by the character 
$l(\vartheta)[k] = d \neq b$. Therefore, $x$ has a least two successors.
\end{proof}

\begin{defi} \label{def:weights}
Let $\phi$ be a node at level $k$ in $\sufftrie(S)$ ($\phi \in \sufftrie[k](S)$). Define
\[
w_k(\phi) =\left\{\begin{array}{ll}
1&\mbox{if $\phi$ has a unique Weiner link $(\phi,\varphi,c)$
and $\varphi$ has no siblings}\\
|\phi| &\mbox{otherwise}
\end{array}\right.
\]
\end{defi}

\begin{lemma} \label{lem-sumOfWeights}
Let $m_k$ be the number of edges in the edge-reduced de Bruijn graph
$\tilde{G}_k(S)$. Then $m_k = \sum_{\phi\in\sufftrie[k](S)} w_k(\phi)$.
\end{lemma}
\begin{proof}
Let $\phi$ be a node in $\sufftrie[k](S)$. The weight
$w_k(\phi)$ equals the number of incoming edges to node $y=l(\phi)$
in $\tilde{G}_k(S)$: If $\phi$ has a unique Weiner link $(\phi,\varphi,c)$
and $\varphi$ has no siblings, then the corresponding edge 
$(x,y)^m$ in $G_k(S)$ is fusible, and its multiplicity is reduced to $1$ 
in $\tilde{G}_k(S)$. Otherwise, all the $|\phi|$ incoming edges to $y$ 
are not fusible.
\end{proof}

Using Lemma \ref{lem-sumOfWeights}, we are able to
describe a naive solution to the edge minimization problem.
We traverse $\sufftrie(S)$ level-wise. At level $k$, we compute the weight 
$w_k(\phi)$ of each node $\phi$ at this level by checking whether or not
$\phi$ has a unique Weiner link $(\phi,\varphi,c)$
and $\varphi$ has no siblings. The sum of these weights yields $m_k$,
the number of edges in $\tilde{G}_k(S)$ and
$\tilde{k} = \argmin_{k \in [1,n]}\{m_k\}$ is the solution we are looking for.
The drawback of this method is that all nodes of $\sufftrie(S)$ must
be visited. Since there are $O(n^2)$ many nodes in $\sufftrie(S)$, the 
resulting worst-case time complexity of this method is $O(n^2)$, 
which is not acceptable for big data.


\Section{Edge minimization algorithm}

To improve the run-time of the naive computational approach, we will 
update the amount of edges level by level, where levels are visited from top
to bottom.

\begin{lemma} \label{lem-WLs}
For every path $p$ from the root to a leaf in $\sufftrie(S)$, there is exactly
one node $\phi$ in $p$ so that (1) $\phi$ has a unique Weiner link, (2) the
ancestors of $\phi$ in $p$ have multiple Weiner links, and (3) every node in 
the subtree rooted at $\phi$ has a unique Weiner link.
\end{lemma}

\begin{proof}
Let $\phi$ be the first node on a path from the root to a leaf $\vartheta$
that has a unique Weiner link ($\vartheta$ has a unique Weiner link, so $\phi$
must exist). Let $c$ be the last character in the label of $\vartheta$,
i.e.\ $l(\vartheta)[n]=c$. Since $\phi$ has a unique Weiner link,
the last character in the label of each leaf in the subtree rooted at $\phi$
must be $c$. This implies (3) and proves the lemma.
\end{proof}

\begin{lemma} \label{lem-oneChild}
Let $\zeta$ be a node in $\sufftrie(S)$ at level $k<n$ that has only one child $\phi$.
\begin{enumerate}[1.\hspace{-2em}]
\item \hspace{2em}If $\zeta$ has a unique Weiner link, then $w_{k+1}(\phi)=1$.
\item \hspace{2em}If $\zeta$ has multiple Weiner links, then 
$w_{k+1}(\phi)=w_k(\zeta)$.
\end{enumerate}
\end{lemma}
\begin{proof}
(1) If $\zeta$ has a unique Weiner link, then $\phi$ has a unique Weiner 
link $(\phi,\varphi,c)$. It is easy to see that in this case
$\phi$ is the only child of its parent if and only if $\varphi$ is the only 
child of its parent. Thus, $\varphi$ has no siblings and $w_{k+1}(\phi)=1$
by Definition \ref{def:weights}.\\
(2) If $\zeta$ has multiple Weiner links, then so does its only child $\phi$.
By Definition \ref{def:weights}, $w_k(\zeta)= |\zeta|$ and 
$w_{k+1}(\phi)= |\phi|$. Now the claim follows from $|\zeta|= |\phi|$.
\end{proof}

According to the preceding lemma, for every node $\phi$ in $\sufftrie(S)$ at level 
$k+1$ that has no siblings, we can infer $w_{k+1}(\phi)$ if we know
the number of Weiner links of its parent $\zeta$ as well as $w_k(\zeta)$.
Moreover, a node $\phi$ that has no siblings \emph{inherits} the weight from
its parent $\zeta$ except for the special case $\ast$ in which $\zeta$ has 
a unique outgoing Weiner link $(\zeta,\psi,c)$ and $\psi$ has 
multiple siblings (in this case $w_k(\zeta)= |\zeta|$ and $w_{k+1}(\phi)=1$).

\begin{algorithm}[t]
\scriptsize
\KwData{Rotation-trie $\sufftrie(S)$ of a string $S$ over alphabet $\Sigma$, a bitvector $F$ of size $|S|$ initialized with zeros.}
\KwResult{Order $k$ such that the edge-reduced deBruijn graph ${\tilde G}_k(S)$ has the minimum number of edges.}
\BlankLine
\BlankLine
\DontPrintSemicolon
\parbox{8em}{$n_T \gets \sigma$} ~ $m \gets |S|$ \tcp*[r]{node counter and edge counter}
\parbox{8em}{$k^* \gets 1$} ~ $m^* \gets |S|$ \tcp*[r]{order with minimal edgecount $m^*$}
\parbox{8em}{$fusible \gets 0$}  \tcp*[r]{inherited fusions}
push the children of the root node of $\sufftrie(S)$ onto a queue $Q$\;
\BlankLine
\For{$k \gets 1 ~ \KwTo ~ |S|-1$}{
	$m \gets m - fusible$ \tcp*[r]{incorporate newly inherited fusions}
	$fusible \gets 0$ \;
	\BlankLine
	\ForEach(\tcp*[f]{reverse old fusions}){$\phi \in Q$}{
		\If{$F[rb(\phi)] = 1$}{
			$\zeta \gets $ parent of $\phi$ \;
			$m \gets m + (|\zeta| - 1)$ \;
			$F[rb(\phi)] \gets 0$ \;
		}
	}
	\BlankLine
	$size \gets |Q|$ \;
	\For{$q \gets 1 ~ \KwTo ~ size$}{
		$\phi \gets $ first element of $Q$ \;
		pop first element of $Q$ \;
		\If(\tcp*[f]{reclassify edges}){$\phi$ has a unique Weiner link $(\phi,\varphi,c)$}{
			\If(\tcp*[f]{fusion in current order}){$\varphi$ has no siblings}{
				$m \gets m - (|\phi| - 1)$\tcp*[f]{$|\varphi|=|\phi|$}
			} \Else(\tcp*[f]{possible fusion in next order}) {
				$fusible \gets fusible + (|\phi| - 1)$\tcp*[f]{$|\varphi|=|\phi|$}
			}
			$F[rb(\phi)] \gets 1$ \;
		}
		\BlankLine
		\ForEach(\tcp*[f]{enqueue new nodes}){outgoing Weiner link $(\phi,\varphi,c)$ of $\phi$}{
			\If{$\varphi$ has siblings}{
				\If(\tcp*[f]{update node counter}){$\varphi$ is not the rightmost sibling}{
					$n_T \gets n_T + 1$ \;
					\If{$n_T \geq m^*$}{
						\Return $k^*$ \;
					}
				}
				push $\varphi$ onto the back of $Q$ \;
			}
		}
	}
	\BlankLine
	\If{$m < m^*$}{
		$k^* \gets k$ \;
		$m^* \gets m$ \;
	}
}
\Return $k^*$ \;

\caption{De Bruijn graph edge minimization using a rotation-trie.}
\label{alg:edgemin-st}
\end{algorithm}

These observations lead to an efficient incremental algorithm, which
calculates $m_k$ from $m_{k-1}$ by a case
distinction on the differences between the edge-reduced de Bruijn graphs
${\tilde G}_{k-1}(S)$ and ${\tilde G}_k(S)$.
The final implementation (Algorithm \ref{alg:edgemin-fm}) does not efficiently
support the operation 'How many children does a node have?'
That is why Algorithm \ref{alg:edgemin-st} does not use this operation.

At the moment, we will ignore lines 25-28 of the algorithm.
Before we will prove the correctness of Algorithm \ref{alg:edgemin-st},
we exemplify the different cases in the 
algorithm by considering the nodes T, TG, TGG, and TGG\$ in the example
of Figure \ref{fig:dbg-st}. In our considerations, we will make use of
the following fact (see Theorem \ref{thm-correctness}):
Before the body of the outer for-loop is executed for a value $k$,
the queue $Q$ contains exactly those nodes of $\sufftrie(S)$ at level $k$ that have 
siblings (nodes that have no sibling are treated implicitly by the
observations above). In the execution of the for-loop
for $k=1$, it is detected that node T has a unique Weiner link to GT, but
GT has siblings. Thus, the multiple edges from G to T in the DBG of order $1$ 
cannot be fused. In line 21, the variable $fusible$ is incremented
by $|T|-1$ ($|T|$ denotes the number of leaves in the subtree rooted at T)
and the rightmost leaf
in the subtree rooted at T is marked (i.e.\ $F[rb(T)]=1$).
The variable $fusible$ will be used to take care of the special case $\ast$.
From that point in time, it is clear that TG, TGG, and TGG\$ have unique Weiner 
links. 

Before the for-loop is executed for $k=2$, the queue $Q$ contains
the nodes AG, G\$, GG, and GT. In line 6, $m=m_1$ 
is decremented by the value of $fusible$. Recall that the contribution of 
node T to $fusible$ was $|T|-1$. Now, what is the contribution of 
node TG to $m_2$? Obviously, TG has a unique Weiner 
link pointing to GTG because T has a unique Weiner link pointing to GT.
Furthermore, TG has siblings if and only if GTG has siblings.
We infer from TG $\notin Q$ that TG has no siblings, hence neither GTG does.
Consequently, $w_2(TG) =1$ because the multiple edges from GT to TG in the
DBG of order $2$ can be fused. The difference between the weights of T and
TG is $w_1(T)-w_2(TG) = |T|-1$. This value was (correctly) subtracted from
$m=m_1$. What would happen if (fictitiously) the node TG were in $Q$?
In this case, GTG would have siblings (since TG has siblings) 
and the subtraction of $|T|-1$ from $m$ must be compensated
by adding $|T|-1$ to $m$.
This would be done in the loop in lines 8 to 12. If the node TG were in $Q$,
then its parent T would have multiple children, and all of its
children would be in $Q$ because they have siblings. For the rightmost child
$\vartheta$ of T, it would follow $rb(\vartheta) = rb(T)$.
Since $F[rb(T)]=1$, the value $w_2(TG)-1$ would be added to $m$ in line 11
(and $F[rb(T)]$ would be set to $0$). 

Before the for-loop is executed for $k=3$, the queue $Q$ contains
the nodes GG\$ and GGT. Node TGG has a unique Weiner 
link pointing to node GTGG because TG has a unique Weiner link pointing to GTG.
The fact that TGG is not in $Q$ means that GTGG has no siblings. Therefore, the
multiple edges from GTG to TGG in the DBG of order $3$ can be fused and 
$w_3(TGG) = 1 = w_2(TG)$. Since $w_2(TG)$ is already included in the value
of $m$, there is no need to change $m$ w.r.t.\ node TGG (i.e.\ given $m_2$,
the incremental computation of $m_3$ does not need to consider this case 
explicitly). We say that TGG (implicitly) inherits its weight from its
parent TG.

Before the for-loop is executed for $k=4$, the queue $Q$ contains
the nodes TGG\$ and TGGT. Node TGG\$ has a unique Weiner link pointing to 
node GTGG\$ because TGG has a unique Weiner link pointing to GTGG.
However, TGG\$ has siblings because it is contained in $Q$. In this case,
$w_3(TGG)$ cannot be inherited and $m$ must be updated. The change
from $m_3$ to $m_4$ w.r.t.\ node TGG and its children 
$\vartheta_1,\dots,\vartheta_q$ can be computed by adding
(a) $|TGG|-1$  to $m$ (this corresponds to the situation in which 
none of the nodes $\vartheta_1,\dots,\vartheta_q$ admits a fusion) and 
(b) subtracting $|\vartheta_i|-1$ from $m$ for each child
$\vartheta_i$ of node TGG that admits a fusion.
(a) will be done in lines 8 to 12 because $\vartheta_q \in Q$
and (b) will be done in the for-loop starting at line 14.
Let us prove the correctness of Algorithm \ref{alg:edgemin-st}.

\begin{theorem}\label{thm-correctness}
Algorithm \ref{alg:edgemin-st} maintains the following invariants:
\begin{enumerate}[1.\hspace{-2em}]
\item \hspace{2em}Before the body of the for-loop starting at line 5 is executed for a value $k$,\newline
	\phantom{3.~~\,}  the queue $Q$ contains exactly those nodes of $\sufftrie[k](S)$ 
that have siblings.\label{invariant1}
\item \hspace{2em}After the body of the for-loop starting at line 5 was executed for a value $k$,\newline
	\phantom{3.~~\,}  the variable $m$ stores the value $m_{k}$.\label{invariant2}
\item \hspace{2em}Before the body of the for-loop starting at line 5 is executed for a value $k$,\newline
	\phantom{3.~~\,}  the variable $n_T$ stores the value $|\sufftrie[k](S)|$.\label{invariant3}
\end{enumerate}
\end{theorem}
\begin{proof}
For $k=1$ property \ref{invariant1} holds true because initially
all nodes of $\sufftrie[1](S)$ (the children of the root node) 
are pushed onto $Q$ (and there are
at least two because of the EOF-symbol). So let $k\geq 2$.
Lines 23--29 of Algorithm \ref{alg:edgemin-st} push a 
node $\varphi$ onto $Q$ only if $\varphi$ has siblings. So we must prove that
\emph{all} nodes with siblings are pushed onto $Q$. For an indirect proof,
suppose that there is a node $\varphi \in \sufftrie[k](S)$ that has
siblings but is not pushed onto $Q$. Let $(\phi,\varphi,c)$ be its
incoming Weiner link. Clearly, if $\varphi$ has siblings, so does $\phi$.
By induction, we may assume that before the
for-loop is executed for the value $k-1$, the queue $Q$ contains 
exactly those nodes of $\sufftrie[k-1](S)$ that have siblings. 
Consequently, at that point in time, $\phi$ is in $Q$. Hence
$\varphi$ will be pushed onto $Q$ in lines 23--29. This contradiction
proves that property \ref{invariant1} is an invariant.

In order to show that property \ref{invariant2} is also invariant,
we may assume by induction
that $m = m_{k-1}$ before the for-loop is executed for $k$ (initially
$m = m_0 = |S|$, i.e.\ there are no fusions). We will show by a case 
distinction that after the execution of the body of the for-loop for $k$, 
the value of $m$ equals $m_k$. Let $\phi$ be a node in $\sufftrie(S)$ at 
level $k$ and let $\zeta$ be its parent node.
\begin{enumerate}
\item $\phi$ has no siblings (or equivalently, $\phi \notin Q$).
In this case, $\phi$ has a unique Weiner link if and only if 
$\zeta$ has a unique Weiner link.
\begin{enumerate}
\item $\phi$ has a unique Weiner link $(\phi,\varphi,c)$.\\
In this case, $w_k(\phi)=1$ because $\varphi$ has no siblings. Furthermore,
$\zeta$ has a unique Weiner link to the parent $\psi$ of $\varphi$. If $\psi$
has no siblings, then $w_{k-1}(\zeta)=1$ and $\phi$ inherits its weight from 
its parent. Otherwise, the parent $\psi$ of $\varphi$ has siblings. In this 
special case $\ast$, the variable fusible was incremented by $|\zeta|-1$ 
in the previous iteration of the for-loop, and this value was subtracted 
from $m$ in the current iteration. 
\item $\phi$ has multiple Weiner links.\\
In this case, $\zeta$ has multiple Weiner links. Since
$|\zeta|=|\phi|$, it follows that $w_{k-1}(\zeta) = |\zeta|=|\phi| = w_k(\phi)$.
That is, $\phi$ inherits its weight from its parent.
\end{enumerate}
\item $\phi$ has siblings (or equivalently, $\phi \in Q$).
\begin{enumerate}
\item If $F[rb(\zeta)]=1$, then there is an ancestor $\alpha$ of $\zeta$
such that (i) $rb(\alpha)=rb(\zeta)$,
(ii) the variable fusible was incremented by 
$|\alpha|-1$ in a previous iteration $i$ of the for-loop,
(iii) $F[rb(\alpha)] = F[rb(\zeta)]$ was set to $1$ in iteration $i$,
(iv) the value $|\alpha|-1$ was subtracted from $m$ in iteration $i+1$,
and (v) $F[rb(\zeta)]=1$ was not reset to $0$ from iteration $i+1$ to the 
current iteration. It follows as a consequence that node $\zeta$ 
inherited its weight from its ancestors in those iterations and therefore 
$|\zeta|=|\alpha|$. Note that $\phi$ and all its siblings are in $Q$.
Since only its rightmost sibling $\vartheta$ is marked
(i.e.\ $F[rb(\vartheta)]=1$), $|\zeta|-1 = |\alpha|-1$
is added once to $m$ in line 11. In summary,
the subtree rooted at $\zeta$ contributes the value $|\zeta|$ to $m$.
This corresponds to the situation in which no child of $\zeta$ admits a 
fusion. Since $\phi$ and all its siblings are
in $Q$, possible fusions of the children of $\zeta$
are taken care of by lines 17-19.
\item If $F[rb(\zeta)]=0$, then the subtree rooted at $\zeta$ contributes the 
value $|\zeta|$ to $m=m_{k-1}$ and this value is not changed in the current
iteration. The remaining arguments are the same as in the preceding
case.
\end{enumerate}
\end{enumerate}
We finally show property \ref{invariant3} holds.
For $k=1$ it is true, because $\sufftrie[1](S)$ contains a node for each letter of $\Sigma$
and $n_T$ is initialized with the value $\sigma$. We further show that $n_T$ will be 
correctly updated in a loop iteration. Therefore we have a look on the changes between 
$|\sufftrie[k](S)|$ and $|\sufftrie[k+1](S)|$. The number of nodes in $\sufftrie[k+1](S)$ equals the 
number of children of the nodes in $\sufftrie[k](S)$. Each node in $\sufftrie[k](S)$ has exactly one
rightmost child (a single child is rightmost). So the amount of nodes changes by the 
number of children that are not rightmost. If a node is a not rightmost child, it has 
siblings and therefore it  will be added to the queue in line 29 due to property \ref{invariant1}. 
Right before that $n_T$ will be incremented by one (line 25-26).
Hence exactly the children that are not rightmost cause an incrementation of $n_T$ 
during the for-loop, the variable stores the correct value at the begin of the next iteration.
\end{proof}

Next, we will explain lines 25-28 of Algorithm \ref{alg:edgemin-st}.
As the node of each edge-reduced DBG must have at least 1 outgoing edge,
the number of edges in such a graph must be greater than or equal to its number of nodes.
This means that there are at least $|\sufftrie[k](S)|$ edges in the edge-reduced DBG of order $k$.
Furthermore, the number of nodes per level is increasing from top to bottom 
($|\sufftrie[k](S)| \leq |\sufftrie[k+1](S)|$ for all $k$). When it is clear
that the next level will contain $n_T \geq m^*$ nodes, the algorithm terminates
and returns $k^*$. This is correct because the current level then must
have at least $n_T \geq m^*$ edges (each node in the next level has at least one incoming
Weiner link) and each further level must have at least $n_T$ nodes and thus at least $m^*$ edges.

\begin{theorem}\label{thm-time complexity}
Given the rotation-trie $\sufftrie(S)$ and a zero-initialized bitvector $F$,
Algorithm \ref{alg:edgemin-st} has a worst-case time complexity of $O(m^*)$.
\end{theorem}
\begin{proof}
We will use an amortized analysis to show the theorem. It is assumed
that the rotation-trie $\sufftrie(S)$ is given and that the following 
operations are supported in constant time:
\begin{itemize}
\item determine the (size of the) parent of a node (lines 10--11)
\item generate Weiner links (lines 17 and 23)
\item check whether a node $\varphi$ has siblings (lines 18 and 24): 
\item if a node $\varphi$ has siblings, check whether $\varphi$ is
the rightmost sibling (line 25)
\end{itemize}
Let $M$ be the set of all nodes $\varphi \in \sufftrie(S)$ such that
there is a Weiner link $(\phi, \varphi, c)$ generated by 
Algorithm \ref{alg:edgemin-st}.
It is clear that the run-time of the algorithm is proportional to $|M|$.
Let $m^*$ and $k$ be the values of the variables upon termination of 
Algorithm \ref{alg:edgemin-st}. We will show that $|M|\leq 4m^*$.
It is clear that $M \subseteq \sufftrie[\leq k+1](S)$, because
the algorithm terminates in the $k$-th iteration (in which nodes of 
the $k+1$ level of $\sufftrie(S)$ are generated). We partition $M$ into four 
disjoint subsets $M_1,\dots,M_4$ and show that each subset 
is of size $\leq m^*$. The definition of the subsets is depicted in 
Table \ref{partition}.
\begin{table}
\begin{tabular}{|c|c|c|c|}
\hline 
$\varphi \in M$ & $\varphi$ has a right sibling & $\varphi$ is rightmost sibling & $\varphi$ has no sibling \\
\hline 
$\varphi \in M \cap \sufftrie[\leq k](S)$ & \multirow{2}{*}{$M_1$} & $M_2$ & $M_4$ \\ 
\cline{1-1}\cline{3-4}
$\varphi \in M \cap \sufftrie[k+1](S)$ &  & \multicolumn{2}{|c|}{$M_3$} \\ 
\hline 
\end{tabular}
\caption{partition of $M$}
\label{partition}
\end{table}
The size of $M_1$ is exactly $m^*$ because of lines 25-28 of 
Algorithm \ref{alg:edgemin-st}. Each node in $M_2$ has at least one left 
sibling and this sibling is in $M_1$, hence $|M_2| \leq |M_1| = m^*$.
Next, we show that $|M_3| \leq m^*$.
Since $n_T = |\sufftrie[k](S)|$ at the end of the $(k-1)$th iteration
and the algorithm did not terminate in this iteration (i.e.\ $n_T < m^*$), 
it follows that $|\sufftrie[k](S)| < m^*$. Each node in $\sufftrie[k](S)$ has 
exactly one child $\varphi$ that has no right sibling (either $\varphi$
has no siblings or $\varphi$ is the rightmost sibling). So the set 
$\{\varphi \mid \varphi \in \sufftrie[k+1](S) \text{ and } \varphi \text{ has no right sibling}\}$ has a size $< m^*$. Since this is a superset of $M_3$, it 
follows that $|M_3| < m^*$.

Finally, we show that $|M_4| \leq |\sufftrie[k](S)|$ (the reader should be aware
of the fact that this is by no means obvious). The theorem then follows from
$|\sufftrie[k](S)| < m^*$, which was shown above. With each node $\varphi \in M_4$,
we associate the lexicographically smallest node $\theta\in \sufftrie[k](S)$
such that $l(\varphi)$ is a suffix of $l(\theta)$. We claim that if 
$\varphi \in M_4$ is associated with $\theta\in \sufftrie[k](S)$, then there
is no other node $\varphi'\in M_4$ 
that is associated with $\theta$. 
To prove the claim, suppose to the contrary that there are two 
distinct nodes $\varphi,\varphi'\in M_4$ which are associated with $\theta$.
So both $l(\varphi)$ and $l(\varphi')$
are suffixes of $l(\theta)$. W.l.o.g.\ we may assume that
$|l(\varphi)| < |l(\varphi')|$ because $\varphi \neq \varphi'$.
It follows as a consequence that $l(\varphi)$ is a proper suffix of 
$l(\varphi')$, hence $l(\varphi') =  c^q\dots c^1l(\varphi)$ for a $q\geq 1$ 
and $c^1,\dots ,c^q \in \Sigma$. Let $\varphi_i$ be the node with label
$c^i\dots c^1l(\varphi)$ for $1\leq i \leq q$. Note that $\varphi'=\varphi_q$.
Since $\varphi$ has no sibling, 
it will not be pushed onto the queue $Q$. Therefore, the Weiner link 
$(\varphi,\varphi_1,c_1)$ will not be generated and thus $\varphi_1$ will not 
be pushed onto the queue $Q$ either. By iterating this argument, we 
conclude that the Weiner link $(\varphi_{q-1},\varphi_q,c_q)$ will not 
be generated. So $\varphi' \notin M$ and this implies $\varphi' \notin M_4$.
This contradiction shows that distinct nodes $\varphi,\varphi'\in M_4$ 
are associated with distinct nodes $\theta',\theta\in \sufftrie[k](S)$.
Since each node $\varphi \in M_4$ is associated with a node
$\theta\in \sufftrie[k](S)$, it follows $|M_4| \leq |\sufftrie[k](S)|$. 
\end{proof}


\Section{Implementation with an FM-index}
To implement Algorithm \ref{alg:edgemin-st}, we need a framework that (1)
allows us to traverse a rotation-trie level by level in a top-down fashion and 
(2) supports the operations listed in the preceding section (in the
run-time analysis of Algorithm \ref{alg:edgemin-st}).
The work of Beller et al.\ \cite{BEL:BER:OHL:2012} provides such a framework
(although this is non-obvious at first sight). To be precise, they developed 
an algorithm that enumerates all LCP-intervals in a top-down manner; see also
\cite[Algorithm 7.19]{OHL:2013}. Their algorithm uses the wavelet tree of
the Burrows-Wheeler transform of string $S$ \cite{BUR:WHE:1994}.

Recall that the Burrows-Wheeler transform computes the string $\LCol$ for a 
null-terminated string $S$ of length $n$ as follows:
\begin{enumerate}
\item Form a conceptual matrix $M'$ whose rows are the rotations (cyclic shifts)
  of $S$.
\item Compute the matrix $M$ by sorting the rows of $M'$ lexicographically. 
\item Output the last column $\LCol$ of $M$.
\end{enumerate}
The Burrows-Wheeler transform can be generalized to multiple strings.

If $\omega$ is a substring of $S$, then the $\omega$-interval is the largest
interval $[i,j]$ of $M$ such that $\omega$ is a prefix of the rotations of $S$
in the interval $[i,j]$. In what follows, for each character $c$, 
$C[c]$ is the overall number of occurrences of 
characters in $\BWT[1,n]$ that are strictly smaller than $c$. 
Given the $\omega$-interval $[lb,rb]$ and a character $c$, 
the $c\omega$-interval $[i,j]$ can be computed 
by $i= C[c] + rank_c(\BWT,lb-1)+1$ and $j= C[c] + rank_c(\BWT,rb)$,
where $rank_c(\BWT,lb-1)$ returns the number of occurrences of character 
$c$ in the prefix $\BWT[1,lb-1]$
(we have $i\leq j$ if $c\omega$ is a substring of $S$; otherwise $i > j$);
see \cite{FER:MAN:2000} for details.
The (balanced) \emph{wavelet tree} \cite{GRO:GUP:VIT:2003} of $\BWT$ 
supports such a backward search step in $O(\log \sigma)$ time, where $\sigma$
is the size of the alphabet.
Backward search can be generalized on the wavelet tree as 
follows: Given an $\omega$-interval $[lb,rb]$, a slight modification of
the procedure $\getintervals(lb,rb)$ described in \cite{BEL:GOG:OHL:SCH:2013}
returns the list $[(c,[i,j]) \mid c\omega \mbox{ is a substring of } S
\mbox{ and } [i,j] \mbox{ is the }$ $c\omega\mbox{-interval}]$, 
where the first component of an element $(c,[i,j])$ is a character.
The worst-case time complexity of the procedure $\getintervals$ is
$O(occ + occ \cdot \log (\sigma/occ))$, where $occ$ is the number of elements 
in the output list; see \cite[Lemma 3]{GAG:NAV:PUG:2012}. Thus, the time 
per generated interval is $O(\log \sigma)$.

In Algorithm \ref{alg:edgemin-fm}, a node $\phi$ is represented by 
the interval $[lb(\phi),rb(\phi)]$. Since for each $k$ the intervals of 
all nodes in $\sufftrie[k](S)$ (the nodes of level $k$) are a partition 
of the interval $[1,n]$, we can represent $\sufftrie[k](S)$ by 
marking all the right boundaries of the node intervals in a bitvector of size $n$. 
For technical reasons, we use a bitvector $B$ of size $n + 1$ in
Algorithm \ref{alg:edgemin-fm}, in which $B[1]$ and $B[n+1]$ are always set to
$1$. 
Given a level $k$, $B[rb(\phi)+1]$ is set to 1 for all nodes $\phi \in \sufftrie[k](S)$. 
Fix some node $\zeta \in \sufftrie[k-1](S)$ and let $\phi$ $\in \sufftrie[k](S)$ be its 
rightmost child. Since $rb(\zeta)=rb(\phi)$, $B[rb(\zeta)+1]=1$ must hold at both levels 
$k-1$ and $k$. In other words, the 
representation of the next higher level can be obtained by just adding ones 
for the boundaries of new nodes (line 14). Algorithm \ref{alg:edgemin-fm} 
maintains the following invariant: when the body of the for-loop
starting at line 8 is executed for a value $k$ and line 23 is reached,
then the boundaries of all nodes in level $k$ are marked in $B$.

Given a node $\phi$ in level $k$, 
all its Weiner links $(\phi,\varphi,c)$ can be generated by
the procedure call $\getintervals(lb(\phi),rb(\phi))$ because 
$[lb(\phi),rb(\phi)]$ is the $l(\phi)$-interval and the call
returns the set
$M=\{ (c,[i,j]) \mid c\in \Sigma \mbox{ and } [i,j] \mbox{ is the $cl(\phi)$-interval}\}$.
The set $M$ contains just one element if and only if there is a 
unique Weiner link $(\phi,\varphi,c)$.  The following statements hold true:
\begin{enumerate}
\item Node $\varphi$, represented by the interval $[lb(\varphi),rb(\varphi)]$,
has no siblings if and only if $B[lb(\varphi)]=1$ and $B[rb(\varphi)+1]=1$.
\item If $\varphi$ has a sibling, then $\varphi$ is the rightmost sibling
if and only if $B[rb(\varphi)+1]=1$.
\end{enumerate}
To prove (1), consider the parent node $\psi$ of $\varphi$. Since
$\psi$ is a node at level $k$, we have $B[lb(\psi)]=1$, $B[rb(\psi)+1]=1$,
and $B[p]=0$ for all $p$ with $lb(\psi) < p < rb(\psi)+1$. Furthermore,
$lb(\psi) \leq lb(\varphi) < rb(\varphi)+1 \leq rb(\psi)+1$ because
$\varphi$ is a child of $\psi$. If $\varphi$ has no siblings, then
$lb(\varphi) = lb(\psi)$ and $rb(\varphi) = rb(\psi)$, which implies
$B[lb(\varphi)]=1$ and $B[rb(\varphi)+1]=1$. If $\varphi$ has siblings, then
$lb(\varphi) \neq lb(\psi)$ or $rb(\varphi) \neq rb(\psi)$.
Hence $B[lb(\varphi)]=0$ or $B[rb(\varphi)+1]=0$. The simple proof of 
(2) is left to the reader.

Algorithm \ref{alg:edgemin-fm} must also be able to determine the parent of 
a node. That is why it traverses the nodes of a level in 
lexicographical order (i.e.\ from left to right). To this end, it
employs $\sigma$ queues (a queue $Q_c$ for each $c\in \Sigma$).
Algorithm \ref{alg:edgemin-fm} maintains the following invariant:
before the body of the for-loop starting at line 8 is executed for a value $k$,
the queue $Q_c$ stores the intervals $[lb(\phi),rb(\phi)]$ 
of all nodes $\phi \in \sufftrie[k](S)$, that have siblings and whose labels starts with $c$, 
in lexicographical order. 
This is certainly true for $k=1$ because of the initialization of the queues 
in lines 6--7: the procedure call $\getintervals(1,n)$ returns the set 
$\{ (c,[i,j]) \mid c\in \Sigma \mbox{ and } [i,j] \mbox{ is the $c$-interval}\}$ 
and each queue contains just one element. Moreover, if it is for $k$,
it remains true for $k+1$ because the foreach-loop starting at line 25 
processes the characters of $\Sigma$ in alphabetical order.
Since a queue $Q_c$ contains only nodes with siblings and these are in 
lexicographical order, the leftmost child of a node is encountered
first in the foreach-loop starting at line 11, followed by some or no 
'middle children', followed by the rightmost child. 
When the first child is encountered, its left boundary, which is also the left 
boundary of its parent, is stored in the variable $last$ (line 17).
The right boundary of the parent is identical to right boundary of the 
rightmost child. Therefore, the rightmost child $\phi$ is reached
when $B[rb(\phi) +1]=1$ (line 18).
At this moment, we know the left and right boundary of the parent, 
which gives us its size (line 20). In line 22, the variable $last$ is reset 
(to ensure that it will be correctly set to a
new left boundary, when the next leftmost child is encountered).

It follows as in the run-time analysis of Algorithm \ref{alg:edgemin-st}
that, given an FM-Index and two zero-initialized bitvectors, Algorithm \ref{alg:edgemin-fm} has a worst-case run-time of 
$\mathcal{O}(m^*\log \sigma)$. This improves the complete run-time from $O(n^2)$ (construction of the
rotations-trie and Algorithm \ref{alg:edgemin-st}) to $O(n \log \sigma)$ (construction of the FM-index and Algorithm \ref{alg:edgemin-fm}).

\newpage

\begin{algorithm}[H]
\scriptsize
\KwData{FM-index of a string $S$ with alphabet $\Sigma$ supporting the $\getintervals$-function, $\C$-array of $S$,
	bitvector $B$ of size $|S|+1$ initialized with zeros, bitvector $F$ of size $|S|$ initialized with zeros.}
\KwResult{Order $k$ such that the edge-reduced de Bruijn graph ${\tilde G}_k(S)$ has the minimum number of edges.}
\BlankLine
\BlankLine
\DontPrintSemicolon
\parbox{8em}{$B[1] \gets 1$} ~ $B[|S|+1] \gets 1$ \tcp*[r]{boundaries of the root node}
\parbox{8em}{$n_T \gets \sigma$} ~ $m \gets |S|$ \tcp*[r]{node counter and edge counter}
\parbox{8em}{$k^* \gets 1$} ~ $m^* \gets |S|$ \tcp*[r]{order with minimum number of edges $m^*$}
$fusible \gets 0$  \tcp*[r]{inherited fusions}
$last \leftarrow \perp$
\BlankLine
\ForEach{$(c,[i,j]) \in \getintervals(1,|S|)$}{
	initialize a queue $Q_c$ with element $[i,j]$
}
\BlankLine
\For{$k \gets 1 ~ \KwTo ~ |S|-1$}{
	$m \gets m - fusible$ \tcp*[r]{establish new inherited fusions}
	$fusible \gets 0$ \;
	\BlankLine
	\ForEach(\tcp*[f]{reverse old fusions and mark boundaries}){$c \in \Sigma$}{
		\ForEach{$[lb,rb] \in Q_c$}{
			\If{$B[rb+1] = 0$}{
				$B[rb+1] \gets 1$ \;
				\textcolor{red}{write $rb+1$ to an external buffer} \;
				\If(\tcp*[f]{save left boundary of parent node}){$last = \perp$}{
					$last \gets lb$ \;
				}
			}
			\Else {
				\If{$F[rb] = 1$}{
					$m \gets m + (rb - last)$ \;
					$F[rb] \gets 0$ \;
				}
				$last \gets \perp$
			}
		}
	}
	\BlankLine
	\ForEach{$c \in \Sigma$}{
		$size[c] \gets |Q_c|$ \;
	}
	\BlankLine
	\ForEach(\tcp*[f]{in alphabetical order}){$c \in \Sigma$}{
		\For{$q \gets 1 ~ \KwTo ~ size[c]$}{
			$[lb,rb] \gets $ first element of $Q_c$ \;
			pop first element of $Q_c$ \;
			\BlankLine
			$M \gets \getintervals( lb, rb )$ \;
			\If(\tcp*[f]{reclassify edges}){$M$ contains only one element $(c,[i,j])$}{
				\If(\tcp*[f]{node has no siblings}){$B[i] = 1 $ and $B[j+1] = 1$}{
					$m \gets m - (rb - lb)$ \;
				} \Else(\tcp*[f]{possible fusion in next order}) {
					$fusible \gets fusible + (rb - lb)$
				}
				$F[rb] \gets 1$ \;
			}
			\BlankLine
			\ForEach(\tcp*[f]{enqueue new nodes}){$(c,[i,j]) \in M$}{
				\If(\tcp*[f]{node has siblings}){$B[i] = 0 $ or $ B[j+1] = 0$}{
					\If(\tcp*[f]{not rightmost sibling, update node counter}){$B[j+1] = 0$}{
						$n_T \gets n_T + 1$ \;
						\If{$n_T \geq m^*$}{
							\Return $k^*$ \;
						}
					}
					push $[i, j]$ onto the back of $Q_c$ \;
				}
			}
		}
	}
	\BlankLine
	\If{$m < m^*$}{
		$k^* \gets k$ \;
		$m^* \gets m$ \;
		\textcolor{red}{clear the external buffer} \;
	}
}
\Return $k^*$ \;

\caption{De Bruijn graph edge minimization using a FM index.}
\label{alg:edgemin-fm}
\end{algorithm}


\SubSection{Retaining the node boundaries of a solution}

So far we are able to compute the order $k$ such that the edge-reduced DBG ${\tilde G}_k(S)$
contains the minimum amount of edges. For typical
applications, like e.g.\ compressed DBGs, we would like to have some information about the nodes
themselves. This information typically consists of a bitvector $B$ where the boundaries of each node
in the corresponding level in the BWT are marked. We could of course compute these boundaries by another
breadth-first traversal of the trie until level $k$, but this seems unnecessary as Algorithm
\ref{alg:edgemin-fm} already performs kind of a breadth-first traversal.

Indeed, note that Algorithm \ref{alg:edgemin-fm} computes the extra information during its execution
by marking the left boundaries of each node. Unfortunately, once the best level of the trie has been processed,
Algorithm \ref{alg:edgemin-fm} writes additional markings in $B$ until the algorithm terminates. To
retain the node boundaries of the solution, we somehow have to get rid of the additional markings. This is the
point where the red marked commands of Algorithm \ref{alg:edgemin-fm} (lines 15 and 46) come into play:
Once we have processed the best level, line 43 is no longer executed, so the external buffer contains positions of
all ``unwanted'' markings in $B$. Thus, to obtain the node boundaries of the solution, we just have to clear
all markings at positions in the external buffer after Algorithm \ref{alg:edgemin-fm} terminates. As Algorithm
\ref{alg:edgemin-st} has an output-sensitive run-time of $O(m^*)$, the external buffer has a
size of at most $O(m^*)$, and therefore computing the node boundaries requires $O(m^*)$ time.


\Section{Connections to tunneling}

To give an application of the edge minimization approach in de Bruijn graphs, we show
connections to a new BWT compression technique called ``tunneling''.The idea of tunneling is that,
given adjacent rotations in $M$ preceded by equal strings, those equal preceding strings
can be ``tunneled'', i.e.\ fuse the equal strings
to just one string and thereby reduce the size of the BWT. The technique was presented in
\cite{BAI:2018} and extended to Wheeler graphs in \cite{ALA:GAG:NAV:BEN:2019}. The 
technique relies on equal preceding strings, which we will call prefix intervals.%
\footnote{Older publications use the term block, but the term prefix interval seems
	clearer to us because a prefix interval is an interval with equal preceding strings.}

\begin{defi} \label{def:lfmapping}
Let $S$ be a string of length $n$, $\LCol$ be its corresponding BWT. The LF-mapping is a
permutation of integers in range $[1,n]$ defined as follows:
\[ \LF[i] \coloneqq \C[\LCol[i]] + \rank_{\LCol[i]}(\LCol, i) \]
We write $\LF^x[i]$ for the x-fold application of $\LF$, i.e.\,
${\LF^x[i] \coloneqq \underbrace{\LF[\LF[\dotsm\LF[}_{x \text{ times}}i]\dotsm]]}$,
\\[-\baselineskip]define $\LF^0[i] \coloneqq i$ and the inverse of $\LF$ as $\LF^{-1}$.
\end{defi}

\begin{defi} \label{def:prefixintervals}
Let $S$ be a string of length $n$ and let $\LCol$ be its BWT. A prefix interval is a
pair of an integer $w \geq 0$ and an interval $[i,j] \subseteq [1,n]$ (a $\langle w, [i,j] \rangle$ prefix interval) such that
\[
	\LCol[\LF^x[i]] = \LCol[\LF^x[i+1]] = \dotsm = \LCol[\LF^x[j]] \text{ for all } -1 \leq x < w
\]
\end{defi}

Definition \ref{def:prefixintervals} is a combination of definitions from \cite{BAI:2018} and
\cite{ALA:GAG:NAV:BEN:2019}, because both together unleash the full power of tunneling a BWT.
Also note that varying $x$ from $-1$ to $w-1$ ensures that lexicographically adjacent rotations
begin with the same character. Next, we will specify what the fusion of a prefix interval means.

\begin{defi} \label{def:tunneling}
Let $S$ be a string of length $n$ and let $\LCol$ be its BWT. Furthermore, let $\langle w_1,[i_1,j_1]\rangle,
\dotsm,\langle w_k,[i_k,j_k]\rangle$ be a set of disjoint prefix intervals in $\LCol$, that is,
the sets $N_l \coloneqq \{ ~ \LF^x[y] ~ \mid ~ i_l \leq y \leq j_l, 0 \leq x \leq w_l ~ \}$ are pairwise disjoint.

Given two bitvectors $\dout$ and $\din$ of size $n$ initialized with ones, tunneling terms the following procedure:
\begin{enumerate}[1.\hspace{-2em}]
\item	\hspace{2em}For all $1 \leq l \leq k$, set \newline
	\phantom{3.~~~~}$\din[i_{l}+1] = \dotsm = \din[j_{l}] = 0$ and
	$\dout[\LF^{w_l}[ i_{l} + 1 ]] = \dotsm = \dout[\LF^{w_l}[ j_{l} ]] = 0$.
\item	\hspace{2em}For all $1 \leq l \leq k$, mark \newline
	\phantom{3.~~~~}$\LCol[ \LF^x[y] ]$, $\dout[ \LF^x[y] ]$ and $\din[ \LF^{x+1}[y] ]$ for all $i_l < y \leq j_l$ and $0 \leq x < w_l$.
\item	\hspace{2em}\label{enumitem:tunnelingremoval}
	remove all marked entries from $\LCol$, $\dout$ and $\din$.
\end{enumerate}
\end{defi}

Tunneling indeed ``fuses'' the strings of a prefix interval: only the uppermost string (uppermost row) of
the prefix interval remains unchanged, all other strings are removed by item \ref{enumitem:tunnelingremoval}
of Definition \ref{def:tunneling},
see Figure \ref{fig:dbg-bwt-fusion} for an example. Analogously to the LF-mapping in a normal BWT,
the generalized LF-mapping in a tunneled BWT is given by the formula
\[
	i \gets \select_{1}\left( \dout, \rank_{1}\left( \din, \C[\LCol[i]] + \rank_{\LCol[i]}( \LCol, i ) \right) \right)\text{.}
\]
Navigating in a tunneled BWT has one more specialty: if the start of a tunnel is reached, a tunnel entry offset has to be saved,
which is used at the end of the tunnel to jump back to the correct entry, see \cite{BAI:2018} and \cite{ALA:GAG:NAV:BEN:2019}
for more details.

We next want to establish a connection between edge-reduced de Bruijn graphs and tunneling. It shows
 that the paths of nodes with fusible edges in a DBG coincide with prefix intervals in a BWT.

\begin{theorem} \label{thm:dbg-fps-bwt-pis}
Let $G_k(S) = (\kmers, E)$ be a DBG of some string $S$, and let $\LCol$ be the BWT of $S$.
Let $p_1 = (x_{1,1}, x_{1,2}, \ldots, x_{1,w_1}), \ldots,p_m = (x_{m,1},\ldots,x_{m,w_x})$ be all
longest paths in $G_k(S)$ such that
$x_{i,j+1}$  is the only successor of  $x_{i,j}$  and  $x_{i,j}$  is the only predecessor of  $x_{i,j+1}$ for all $1 \leq i \leq m$ and $1 \leq j < w_i$.

\begin{enumerate}[1.\hspace{-2em}]
\item	\hspace{2em}Each path $p=(x_1,\ldots,x_w)$ corresponds to a prefix interval $\langle w-1,[i,j]\rangle$ in $\LCol$ \newline
	\phantom{2.~~\,} where the number of edges between nodes in $p$ equals $(w-1)\cdot( j - i + 1)$.
\item	\hspace{2em}The corresponding prefix intervals are disjoint.
\end{enumerate}
\end{theorem}
\begin{proof}~

\begin{enumerate}
\item	Each node $x$ of a path spans a consecutive $x$-interval $[i,j]$ in the BWT matrix $M$, as rotations prefixed with $x$
	are adjacent in the matrix $M$. Let $x$ and $y$ be two adjacent nodes of $p$ with intervals
	$[i_x,j_x]$ and $[i_y,j_y]$. Because $x$ is the only predecessor of $y$ and $y$ is the only successor of $x$,
	the intervals have the same height ($j_x - i_x = j_y - i_y$) and the entries in the BWT interval of the successor
	$y$ are equal ($\LCol[i_y] = \ldots = \LCol[j_y]$), so $[\LF[i_y],\LF[j_y]] = [i_x,j_x]$ must hold.
	Because each edge from $x$ to $y$ corresponds to a text position where the $k$-mers overlap, there must be exactly
	$j_y-i_y$ edges between $x$ and $y$, so the path overall contains $(w-1) \cdot (j - i + 1)$ edges.
\item	As the paths are disjoint (otherwise one node would have multiple predecessors or successors),
	the intervals in $M$ and therefore the prefix intervals must be disjoint, too.
\end{enumerate}
\end{proof}

\begin{figure}[t]
\scriptsize\centering\hfill%
\begin{tikzpicture}[y=4ex,x=2em,
	one/.style={minimum width=2em},
	two/.style={minimum width=4.8em},
	six/.style={minimum width=15.5em},
	dbgnode/.style={draw, rounded corners, minimum height=1em, minimum width=3em, fill=blue!10,inner sep=0pt},
	link/.style={-latex,black!50,dotted},
	mlink/.style={-latex,very thick,teal!30},
	mlink2/.style={-latex,very thick,teal!70}]
\node at (0,0) {DBG};
\foreach \kmer/\off/\occ [count=\i from 1,evaluate=\occ as \height using 4*\occ - 1] in {
	\$A/1/1,
	AG/2/1,
	G\$/3/1,
	GG/4.5/2,
	GT/6.5/2,
	TG/8.5/2} {
	\node[dbgnode,minimum height=\height ex] (dbgn\i) at (0,-\off) {\kmer};
};
\begin{scope}[on background layer]
\foreach \src/\dest/\extrastyle in {
	1/2/link,
	2/7/link,
	3/1/link,
	4/3/link,
	5/6/link,
	6/8/mlink,
	7/9/mlink,
	8/4/mlink2,
	9/5/mlink2} {
	\draw[-latex,\extrastyle] (0.75,-\src) -- (1.5,-\src) to[out=0,in=180] (-1.5,-\dest) -- (-0.75,-\dest);
};
\end{scope}

\node[anchor=east] at (4,0) {$\LCol$};
\node[anchor=west] at (5,0) {Rotations};

\foreach \bwentry/\kmer/\remsuff [count=\i from 1] in {
	G/\$A/GTG,
	\$/AG/TGG,
	G/G\$/AGT,
	T/GG/\$AG,
	T/GG/TGG,
	G/GT/GG\$,
	A/GT/GGT,
	G/TG/G\$A,
	G/TG/GTG} {
	\node[fill=white,minimum height=4ex,anchor=east] (bwt\i) at (4,-\i) {\bwentry};
	\node[anchor=west] (suff\i) at (5,-\i) {\kmer\textcolor{black!50}{\remsuff$\,\dotsm$}};
	\draw[-latex] (suff\i) -- (bwt\i);
}
\begin{scope}[on background layer]
\draw[-latex,mlink] (bwt8.west) -- ($(bwt8.west) - (.75,0)$) to[out=180,in=210] ($(suff6.south west) + (0,4pt)$);
\draw[-latex,mlink] (bwt9.west) -- ($(bwt9.west) - (.75,0)$) to[out=180,in=210] ($(suff7.south west) + (0,4pt)$);
\draw[-latex,mlink2] (bwt4.west) -- ($(bwt4.west) - (.75,0)$) to[out=180,in=150] ($(suff8.north west) - (0,4pt)$);
\draw[-latex,mlink2] (bwt5.west) -- ($(bwt5.west) - (.75,0)$) to[out=180,in=150] ($(suff9.north west) - (0,4pt)$);
\end{scope}
\end{tikzpicture}%
\begin{tikzpicture}[y=4ex,x=2em,
	one/.style={minimum width=2em},
	two/.style={minimum width=4.8em},
	six/.style={minimum width=15.5em},
	dbgnode/.style={draw, rounded corners, minimum height=1em, minimum width=3em, fill=blue!10,inner sep=0pt},
	link/.style={-latex,black!50,dotted},
	mlink/.style={-latex,ultra thick,teal!30},
	mlink2/.style={-latex,ultra thick,teal!70}]
\node at (0,0) {DBG};
\foreach \kmer/\off/\occ [count=\i from 1,evaluate=\occ as \height using 4*\occ - 1] in {
	\$A/1/1,
	AG/2/1,
	G\$/3/1,
	GG/4.5/2,
	GT/6.5/2,
	TG/8.5/2} {
	\node[dbgnode,minimum height=\height ex] (dbgn\i) at (0,-\off) {\kmer};
};
\begin{scope}[on background layer]
\foreach \src/\dest/\extrastyle in {
	1/2/link,
	2/7/link,
	3/1/link,
	4/3/link,
	5/6/link,
	6/8/mlink,
	8/4/mlink2} {
	\draw[-latex,\extrastyle] (0.75,-\src) -- (1.5,-\src) to[out=0,in=180] (-1.5,-\dest) -- (-0.75,-\dest);
};
\end{scope}

\node[anchor=east] at (4,0) {$\LCol$};
\node[anchor=west] at (4,0) {$\dout$};
\node[anchor=east] at (7,0) {$\din$};
\node[anchor=west] at (7,0) {Rotations};

\foreach \bwentry/\kmer/\remsuff/\doutentry/\dinentry [count=\i from 1] in {
	G/\$A/GTG$\,\dotsm$/1/1,
	\$/AG/TGG$\,\dotsm$/1/1,
	G/G\$/AGT$\,\dotsm$/1/1,
	T/GG/\$AG$\,\dotsm$/1/1,
	\phantom{T}/GG/TGG$\,\dotsm$/\phantom{0}/0,
	G/GT/$\,\dotsm$/1/1,
	A/\phantom{GT}//0/\phantom{0},
	G/TG/$\,\dotsm$/1/1,
	\phantom{G}/\phantom{TG}//\phantom{0}/\phantom{0}} {
	\node[fill=white,minimum height=4ex,anchor=east] (bwt\i) at (4,-\i) {\bwentry};
	\node[fill=white,minimum height=4ex,anchor=west] (dout\i) at (4,-\i) {$\mathtt{\doutentry}$};
	\node[minimum height=4ex,anchor=east] (din\i) at (7,-\i) {$\mathtt{\dinentry}$};
	\node[anchor=west] (suff\i) at (7,-\i) {\kmer\textcolor{black!50}{\remsuff}};
}
\foreach \dinentry/\doutentry in {
	1/1,2/2,3/3,4/4,5/4,6/6,6/7,8/8} {
	\draw[-latex] (din\dinentry) -- (dout\doutentry);
};
\begin{scope}[on background layer]
\draw[-latex,mlink] (bwt8.west) -- ($(bwt8.west) - (.75,0)$) to[out=180,in=210] (suff6.south west);
\draw[-latex,mlink2] (bwt4.west) -- ($(bwt4.west) - (.75,0)$) to[out=180,in=150] (suff8.north west);
\end{scope}
\end{tikzpicture}%
\hfill%
\caption{Correspondence between a de Bruijn graph with order $k=2$ and the BWT (left), and correspondence between
	the edge-reduced de Bruijn graph and the tunneled BWT (right) for $S=AGTGGTGG\$$. The left-hand side shows
	teal-colored fusible edges between nodes in the DBG and the corresponding LF-mappings in the BWT.
	The right-hand side shows the fused edges in the DBG, as well as the LF-mapping in the tunneled BWT by
	tunneling the $\langle 2,[4,5]\rangle$ prefix interval corresponding to the fusible path.}
\label{fig:dbg-bwt-fusion}
\end{figure}
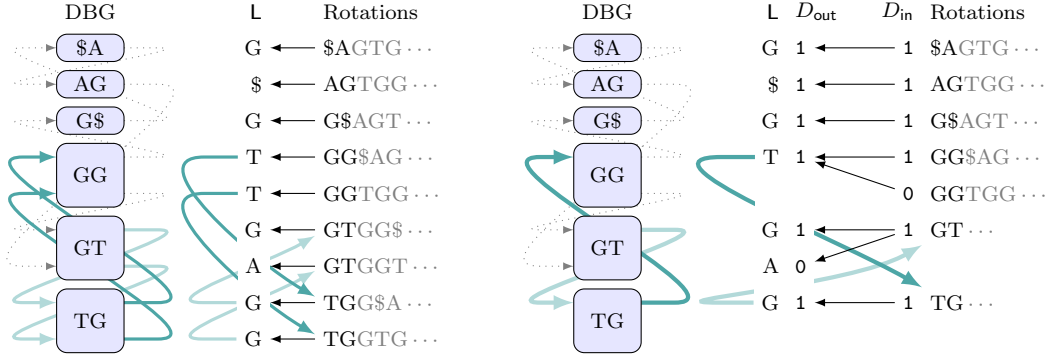

Theorem \ref{thm:dbg-fps-bwt-pis} states that edge-reducible paths in a DBG correspond to
disjoint prefix intervals in the matching BWT, see also Figure \ref{fig:dbg-bwt-fusion}. Moreover,
the length of an edge-reduced DBG and the size of a tunneled BWT, obtained by tunneling the
corresponding prefix intervals, is equal.

\begin{cor} \label{cor:erdbgedgecnttbwtlength}
Let $G_k(S) = (\kmers, E)$ be a DBG of some string $S$, and let $\LCol$ be the BWT of $S$.
Let $p_1 = (x_{1,1}, x_{1,2}, \ldots, x_{1,w_1}), \ldots,p_m = (x_{m,1},\ldots,x_{m,w_x})$ be all
longest paths in $G_k(S)$ such that
$x_{i,j+1}$  is the only successor of  $x_{i,j}$  and  $x_{i,j}$  is the only predecessor of  $x_{i,j+1}$ for all $1 \leq i \leq m$ and $1 \leq j < w_i$,
and let $\langle w_1-1,[i_1,j_1]\rangle,\ldots,\langle w_m-1,[i_m,j_m]\rangle$ be the corresponding prefix intervals in $\LCol$.

Then, the size of the tunneled BWT $\LCol$, $\dout$, $\din$ obtained by tunneling \\
${\langle w_1-1,[i_1,j_1]\rangle,\ldots,\langle w_m-1,[i_m,j_m]\rangle}$ is equal to the number of edges in the edge-reduced de Bruijn graph
$\tilde G_k(S) = (\kmers,\tilde E)$, that is, $|\LCol| = \sum_{(x,y)^m \in \tilde E} m$.
\end{cor}
\begin{proof}
Initially, the number of edges as well as the size of the BWT is $|S|$.
Edge-reducing a path $p=(x_1,\ldots,x_w)$ with corresponding prefix interval $\langle w-1,[i,j]\rangle$
removes exactly $(w-1)\cdot(j-i+1)$ edges from $\tilde G_k(S)$. Analogously, tunneling the prefix interval
$\langle w-1,[i,j]\rangle$ removes exactly $(w-1)\cdot (j - i + 1)$ entries from $\LCol$, $\dout$ and $\din$.
\end{proof}

Summing up, by combining the edge minimization problem with tunneling, we are able to produce a tunneled
BWT which has minimum length under the restriction that only $k$-mer prefix intervals can be tunneled.


\subsection*{Producing a tunneled BWT using the output of the edge minimization algorithm}

After establishing the connection between tunneling and edge minimization in de Bruijn graphs, we
want to describe how Algorithm \ref{alg:edgemin-fm} can be used to produce a tunneled BWT.
It shows that the node boundaries of the $k$-mers of the optimal order $k$ stored in the bitvector $B$
(Section ``Retaining the node boundaries of solution'') suffice to produce the corresponding tunneled BWT.

Looking at Definition \ref{def:tunneling} with the tunneling process, tunneling may be formulated
alternatively as follows: Starting with two one-initialized bitvectors $\din$ and $\dout$,
set the start entries and all intermediate entries of each prefix interval in the bitvector $\din$ to zero,
except for the uppermost row of the prefix interval. Analogously, set the end entries and all intermediate entries
of each prefix interval in $\dout$ to zero, except for the uppermost row of the prefix interval.
Afterwards, remove all entries $\LCol[i]$ and $\dout[i]$ where $\din[i] = 0$ holds,
as well as all $\din[i]$-entries where $\dout[i] = 0$ holds.

It is easy to show that the upper process produces the same tunneled BWT as the normal process from
Definition \ref{def:tunneling}. Also, the process does not regulate the order in which the columns
$[i,j], [\LF[i],\LF[j]], \ldots, [\LF^w[i],\LF^w[j]]$ of a $\langle w,[i,j]\rangle$ prefix interval have
to be marked.

Our approach thus is as follows: Given an interval $[i,j]$ of a $k$-mer $y$ in the BWT (those intervals can be obtained by the
node boundary bitvector), we check if
\begin{enumerate}
\item	$\LCol[i] = \dotsm = \LCol[j]$. In this case, the $k$-mer $y$ has only one predecessor.
\item	$[\LF[i],\LF[j]]$ are the boundaries of another $k$-mer $x$ (this can be checked using the node boundaries).
	In this case, $y$ is the only successor of $x$.
\end{enumerate}
If both conditions are true, we know that edges between $x$ and $y$ can be fused. Combined with the
tunneling process described above, we set $\din[i+1] = \dotsm = \din[j] = 0$ and
$\dout[\LF[i+1]] = \dotsm = \dout[\LF[j]] = 0$. Then, after processing each $k$-mer in this way,
in a front-to-back scan, we remove all entries $\LCol[i]$ and $\dout[i]$ where $\din[i] = 0$ holds,
as well as all $\din[i]$-entries where $\dout[i] = 0$ holds. This clearly yields the requested tunneled BWT
with minimum length. Algorithm \ref{alg:tbwt_from_nodebounds} gathers the described ideas,
and also contains an optimization which allows to build $\din$ and $\dout$ from copies of
the node boundary vector $B$.

\newpage

\begin{algorithm}[H]
\scriptsize
\KwData{FM-index of a string $S$ with alphabet $\Sigma$, BWT $\LCol$ of $S$, $\C$-array of $S$,
	bitvector $B$ of size $|S|+1$ containing the left node bounds of k-mer intervals in the BWT of $S$.}
\KwResult{Tunneled BWT $\LCol$, $\dout$ and $\din$, where prefix intervals corresponding to
	fusible paths in the underlying DBG are tunneled.}
\BlankLine
\BlankLine
\DontPrintSemicolon
rename bitvector $B$ to $\dout$ \;
$\din \gets \text{ copy of } \dout$ \;
$i \gets 1$ \ \tcp*[r]{start of current prefix interval}
\For{$j \gets 1 $ \KwTo $ n$}{
	\If(\tcp*[f]{end of k-mer interval reached}){$\din[j+1] = 1$}{
		$M \gets \getintervals( i, j )$ \;
		\BlankLine
		\tcp*[l]{check if corresponding k-mer has only one predecessor ($\LCol[i] = \dotsm = \LCol[j]$)}
		\tcp*[l]{and if predecessor has only one successor ($\dout[\LF[i]] = \dout[\LF[j]+1] = 1$)}
		\BlankLine
		\If{$M \text{ contains only one element } (c,[lb,rb]) \text{ and } \dout[lb] = 1 \text{ and } \dout[rb + 1] = 1$}{
			\BlankLine
			$\din[i\sdots j] \gets \mathtt{1} \mathtt{0}^{j-i}$ \;
			$\dout[lb\sdots rb] \gets \mathtt{1} \mathtt{0}^{j-i}$
		}
		\Else {
			$\din[i\sdots j] \gets \mathtt{1}^{j-i+1}$ \tcp*[r]{no tunnel start, mark bits in a normal manner}
			\ForEach{$\langle c,[lb,rb]\rangle \in M$}{
				$\dout[lb\sdots rb] \gets \mathtt{1}^{rb-lb+1}$
			}
		}
		$i \gets j+1$
	}
}
\BlankLine
$p \gets 1$ \;
$q \gets 1$ \;
\For(\tcp*[f]{remove redundant entries from $\LCol$, $\dout$ and $\din$}){$i \gets 1 $ \KwTo $ n$}{
	\If{$\din[i] = 1$}{
		\BlankLine
		$\LCol[p] \gets \LCol[i]$ \;
		$\dout[p] \gets \dout[i]$ \;
		$p \gets p + 1$
	}
	\If{$\dout[i] = 1$}{
		\BlankLine
		$\din[q] \gets \din[i]$ \;
		$q \gets q+1$
	}
}
\parbox{8em}{$\dout[p+1] \gets 1$} ~ $\din[q+1] \gets 1$ \tcp*[r]{set termination bits}
trim $\LCol$ to size $p$, and trim $\dout$ and $\din$ to size $p+1$

\caption{Construction of a tunneled BWT using the node boundary bitvector $B$ described in the Section
	``Retaining the node boundaries of solution''. The algorithm requires $O(n + m^* \log \sigma)$
	time using the FM-index in form of a wavelet tree, where $n$ is the length of the normal BWT, $m^*$ is the number
	of edges in the edge-reduced de Bruijn graph and $\sigma$ is the alphabet size.}
\label{alg:tbwt_from_nodebounds}
\end{algorithm}


\Section{Experimental results}
We implemented the described algorithms using \texttt{C++} and the \texttt{sdsl-lite} library
\cite{sdsl-lite-library}. All experiments were conducted on a 64 bit Ubuntu 18.04.3 LTS system equipped
with two 16-core Intel Xeon E5-2698v3 processors and 256 GB of RAM.
All programs were compiled with the O3 option using g++ (version 7.4.0) and are publicly available
\cite{edgemin-impl}. The test data comes from the Pizza \& Chili corpus \cite{pizzachili-corpus}
and the Repetitive corpus \cite{repetitive-corpus}, where nullbytes were removed from the input files
for technical reasons.

Figure \ref{fig:exp-plot} shows our main results belonging to de Bruijn graph edge minimization and
tunneled FM index construction.%
\footnote{We refer to an FM-index as a BWT encoded in a wavelet tree.}
Not too surprising, the edge reduction ratio for non-repetitive files is moderate (5 to 45 \%) but
quite high for repetitive files {(60 - 95 \%)}. Also, the reduction of the data structure size between a
normal and a tunneled FM-index built using the edge minimization algorithm cannot completely keep
up with the edge reduction ratio. This is owed to the two additional bitvectors $\dout$ and $\din$
required to perform a backward step in a tunneled BWT, and also explains why the tunneled FM index for
the test file \texttt{dna} is bigger than the conventional FM-index: DNA consists of a 4 character alphabet.
As a consequence, encoding DNA in a balanced wavelet tree requires roughly 2 bits per symbol.
As a tunneled FM-index requires 2 additional bits per remaining symbol (bitvector $\dout$ and $\din$),
and the edge reduction ratio for \texttt{dna} is low, the additional bitvectors size exceeds the benefits of
reducing the BWT length.

The edge minimization algorithm typically requires half of the time of a full FM-index construction
(suffix array + BWT + wavelet tree), with some variations depending on the number of edges in the edge-minimal DBG of the
underlying data. The algorithm, however, is clearly faster than the naive computational approach, which for the
\texttt{sources} file required more than 8 hours while the fast algorithm required only 43 seconds--we
therefore took a pass on more comparisons with the naive approach. Also, it is surprising that for most files
except for some very repetitive inputs the order of the minimal de Bruijn graph is very small, i.e.\
a number below 30, but we are not aware of a reason for this.

\begin{figure}[t]
\pgfplotstableread{rcrbenchmark.dat}{\bmresults}

\pgfplotstableset{
	create on use/min-edge-percentage/.style={
		create col/expr={\thisrow{min_dbg_edges}*100/\thisrow{input_length}}
	},
	create on use/min-fm-percentage/.style={
		create col/expr={\thisrow{tfm_index_size}*100/\thisrow{fm_index_size}}
	},
	create on use/min-fm-overhead/.style={
		create col/expr={(\thisrow{tfm_index_size}*100/\thisrow{fm_index_size})-(\thisrow{min_dbg_edges}*100/\thisrow{input_length})}
	}
}

\begin{tikzpicture}
\begin{axis}[
	width=.9\textwidth,
	height=40ex,
	grid style=dashed,
	bar width=15pt,
	yticklabel style={},
	legend pos={north east},
	legend style={cells={align=right},legend columns=1},
	reverse legend,
	xticklabels from table={\bmresults}{file},
	xtick=data,
	x tick label style={rotate=90,anchor=east,font=\tiny},
	ybar stacked,
	ylabel={\parbox{11.5em}{\% of original DBG edges respectively FM-index size}}, 
	ytick={0,20,40,60,80,100,120,140}, 
	ymin=-5,
	ymax=150, 
	ymajorgrids=true,   
	extra y ticks=100,
	extra y tick labels={},
	extra y tick style={grid style={black,dashed,very thick}},
    ]   
\addplot+[ybar,draw=blue!50,fill=blue!15] table [y=min-edge-percentage, x expr=\coordindex] {\bmresults};
\addplot+[ybar,draw=teal!50,fill=teal!15] table [y=min-fm-overhead, x expr=\coordindex] {\bmresults};
\legend{\parbox{12em}{edge count in minimal DBG},\parbox{12em}{tunneled FM-index size}}
\end{axis}
\end{tikzpicture}
\vspace{-.5\baselineskip}
\caption{Bar chart for the edge reduction rate of de Bruijn graphs respectively FM-index size reduction using
	tunneling. The y-tick at 100 \% indicates the size of a non-edge-reduced de Bruijn graph resp.\ a
	normal FM-index. The order $k$ that minimizes the edge-reduced de Bruijn graph is depicted right beside
	of each test case name.}
\label{fig:exp-plot}
\end{figure}
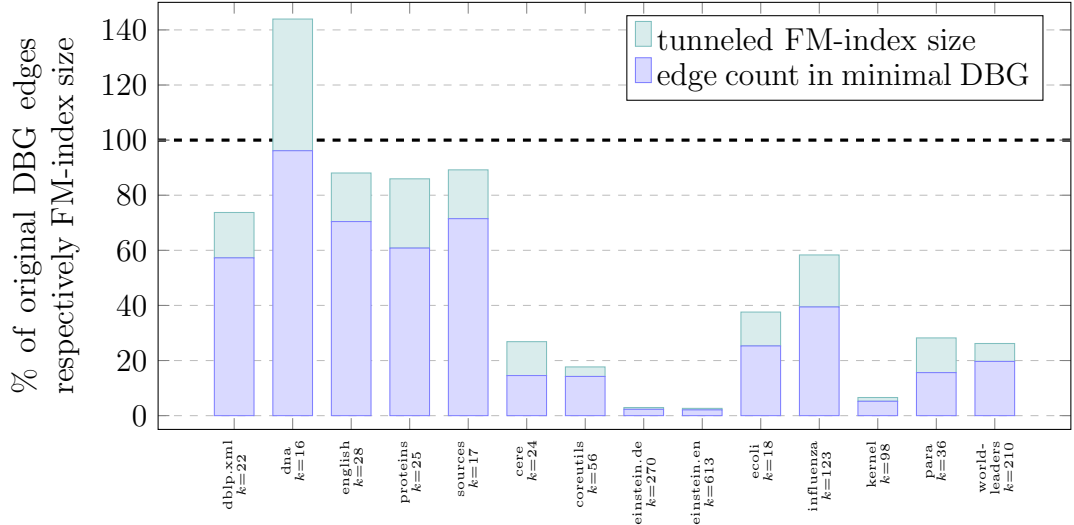


\Section{Conclusion}

We introduced the edge minimization problem and gave an efficient
algorithm to it. The usage of an FM-index improves the total run-time
of $O(n^2)$ for solving the edge minimization problem to $O(n \log \sigma)$ because of
the faster construction of the FM-index compared to the construction of the rotations trie. 
Furthermore, we showed a deep connection between edge
minimization and the ``tunneling'' compression technique. Experiments show that the algorithm
is practical and leads to a significant reduction of edges compared to a non-edge-reduced de Bruijn
graph. Furthermore, despite the required extra information and excluding one test case,
the resulting tunneled FM-index remains smaller than a classical FM-index--especially for
very repetitive inputs, the reduction of the data structure size is about 80 \%.

Therefore, our results are a significant progress on the open problem of finding disjoint prefix intervals in a BWT
\cite{ALA:GAG:NAV:BEN:2019}. They are so in a practical way, but they are also interesting from a theoretical
point of view, although we restrict prefix intervals to the special class of $k$-mer prefix intervals.
Edge-minimal de Bruijn graphs are interesting in their own right because they
constitute graphs with minimum redundancy and may be applicable in other disciplines.

\Section{Acknowledgments}
This work was supported by the DFG (OH 53/7-1).

\Section{References}
\bibliographystyle{IEEEbib}
\bibliography{refs}

\begin{thebibliography}{10}

\bibitem{DEBRU:1946}
Nicolas~G.\ de~Bruijn,
\newblock ``{A Combinatorial Problem},''
\newblock {\em Koninklijke Nederlandse Akademie V.\ Wetenschappe}, vol. 49, pp.
  758--764, 1946.

\bibitem{IDU:WAT:1995}
Ramana~M.\ Idury and Michael~S.\ Waterman,
\newblock ``A new algorithm for {DNA} sequence assembly,''
\newblock {\em Journal of Computational Biology}, vol. 2, no. 2, pp. 291--306,
  1995.

\bibitem{IQB:CAC:TUR:FLI:MCV:2012}
Zamin Iqbal, Mario Caccamo, Isaac Turner, Paul Flicek, and Gil McVean,
\newblock ``De novo assembly and genotyping of variants using colored de bruijn
  graphs,''
\newblock {\em Nature genetics}, vol. 44, pp. 226--32, 2012.

\bibitem{MAR:LEE:SCH:2014}
Shoshana Marcus, Hayan Lee, and Michael Schatz,
\newblock ``Splitmem: A graphical algorithm for pan-genome analysis with suffix
  skips,''
\newblock {\em Bioinformatics}, vol. 30, no. 24, pp. 3476–3483, 2014.

\bibitem{BUR:WHE:1994}
Michael Burrows and David~{J.} Wheeler,
\newblock ``{A block-sorting lossless data compression algorithm},''
\newblock Tech. {R}ep. 124, Digital Equipment Corporation, 1994.

\bibitem{FER:MAN:2005}
Paolo Ferragina and Giovanni Manzini,
\newblock ``{Indexing Compressed Text},''
\newblock {\em Journal of the ACM}, vol. 52, no. 4, pp. 552--581, 2005.

\bibitem{BAI:2018}
Uwe Baier,
\newblock ``{On Undetected Redundancy in the Burrows-Wheeler Transform},''
\newblock in {\em Annual Symposium on Combinatorial Pattern Matching}, 2018,
  CPM '18, pp. 3:1--3:15.

\bibitem{ALA:GAG:NAV:BEN:2019}
Jarno Alanko, Travis Gagie, Gonzalo Navarro, and Louisa~Seelbach Benkner,
\newblock ``{Tunneling on Wheeler Graphs},''
\newblock in {\em Proceedings of the 2019 Data Compression Conference}, 2019,
  DCC '19, pp. 122--131.

\bibitem{edgemin-impl}
Uwe Baier, Thomas B\"{u}chler, Enno Ohlebusch, and Pascal Weber,
\newblock ``{Edge minimization implementation},''
  \url{https://www.uni-ulm.de/fileadmin/website_uni_ulm/iui.inst.190/Forschung/Projekte/seqana/dbg_edge_minimization.tar.bz2}.

\bibitem{pizzachili-corpus}
Paolo Ferragina and Gonzalo Navarro,
\newblock ``{Pizza \& Chili Corpus},''
  \url{http://pizzachili.dcc.uchile.cl/texts.html},
\newblock last visited October 2019.

\bibitem{repetitive-corpus}
Paolo Ferragina and Gonzalo Navarro,
\newblock ``{Repetitive Corpus},''
  \url{http://pizzachili.dcc.uchile.cl/repcorpus.html},
\newblock last visited October 2019.

\bibitem{COM:PEV:TES:2011}
Phillip~E.C.\ Compeau, Pavel~A.\ Pevzner, and Glenn Tesler,
\newblock ``How to apply de {B}ruijn graphs to genome assembly,''
\newblock {\em Nature Biotechnology}, vol. 29, no. 11, pp. 987--991, 2011.

\bibitem{BAI:BEL:OHL:2016}
Uwe Baier, Timo Beller, and Enno Ohlebusch,
\newblock ``{Graphical pan-genome analysis with compressed suffix trees and the
  {B}urrows-{W}heeler transform},''
\newblock {\em Bioinformatics}, vol. 32, no. 4, pp. 497--504, 2016.

\bibitem{FRE:1960}
Edward Fredkin,
\newblock ``Trie memory,''
\newblock {\em Communications of the ACM}, vol. 3, no. 9, pp. 490--499, 1960.

\bibitem{WEI:1973}
Peter Weiner,
\newblock ``Linear pattern matching algorithms,''
\newblock in {\em Proc.\ 14th IEEE Annual Symposium on Switching and Automata
  Theory}, 1973, pp. 1--11.

\bibitem{BEL:BER:OHL:2012}
Timo Beller, Katharina Berger, and Enno Ohlebusch,
\newblock ``{Space-Efficient Computation of Maximal and Supermaximal Repeats in
  Genome Sequences},''
\newblock in {\em Proc.\ 19th International Symposium on String Processing and
  Information Retrieval}, 2012, SPIRE '12, pp. 99--110.

\bibitem{OHL:2013}
Enno Ohlebusch,
\newblock {\em {Bioinformatics Algorithms: Sequence Analysis, Genome
  Rearrangements, and Phylogenetic Reconstruction}},
\newblock Oldenbusch Verlag, 2013.

\bibitem{FER:MAN:2000}
Paolo Ferragina and Giovanni Manzini,
\newblock ``Opportunistic data structures with applications,''
\newblock in {\em Proc.\ 41st Annual {IEEE} Symposium on Foundations of
  Computer Science}, 2000, pp. 390--398.

\bibitem{GRO:GUP:VIT:2003}
Roberto Grossi, Ankur Gupta, and Jeffrey~S.\ Vitter,
\newblock ``High-order entropy-compressed text indexes,''
\newblock in {\em Proc.\ 14th Annual ACM-SIAM Symposium on Discrete
  Algorithms}, 2003, pp. 841--850.

\bibitem{BEL:GOG:OHL:SCH:2013}
Timo Beller, Simon Gog, Enno Ohlebusch, and Thomas Schnattinger,
\newblock ``Computing the longest common prefix array based on the
  {B}urrows-{W}heeler transform,''
\newblock {\em Journal of Discrete Algorithms}, vol. 18, pp. 22--31, 2013.

\bibitem{GAG:NAV:PUG:2012}
Travis Gagie, Gonzalo Navarro, and Simon~J.\ Puglisi,
\newblock ``New algorithms on wavelet trees and applications to information
  retrieval,''
\newblock {\em Theoretical Computer Science}, vol. 426-427, pp. 25--41, 2012.

\bibitem{sdsl-lite-library}
Simon Gog,
\newblock ``{\texttt{sdsl-lite} Library},''
  \url{https://github.com/simongog/sdsl-lite},
\newblock last visited October 2019.

\end{thebibliography}

\end{document}